\RequirePackage{amsmath}
\documentclass[runningheads]{llncs}
\usepackage[T1]{fontenc}
\usepackage{graphicx}

\usepackage[dvipsnames]{xcolor}
\usepackage{tikz}
\usetikzlibrary{positioning}
\usepackage{xfrac}
\usepackage{breqn}
\usepackage{pifont}
\usepackage{tabularx, booktabs}
\usepackage{amsfonts,amssymb}
\usepackage{cleveref}
\usepackage[
  left=1in,
  right=1in,
  top=1in,
  bottom=1in
]{geometry}

\newcolumntype{Y}{>{\centering\arraybackslash}X}
\newcommand{\cmark}{\textcolor{Green}{\ding{51}}}
\newcommand{\xmark}{\textcolor{Red}{\ding{55}}}
\newcommand{\boxmark}{\textcolor{Goldenrod}{\ding{110}}}
\newcommand{\calA}{\mathcal A}
\newcommand{\calF}{\mathcal F}
\newcommand\pavscore{\operatorname{\mathrm{sc}_{\mathrm{PAV}}}}

\begin{document}
\title{Justified Representation: From Hare to Droop}
\titlerunning{From Hare to Droop}

 \author{Matthew M. Casey \and
 Edith Elkind} 

\authorrunning{M. M. Casey and E. Elkind}
 
\institute{Northwestern University, Evanston IL, USA}
%\email{lncs@springer.com}}
%
%\author{}
%\institute{}
\maketitle              % typeset the header of the contribution
\begin{abstract}
The study of proportionality in multiwinner voting with approval ballots has received much attention in recent years~\cite{lackner2023multi}. Typically, proportionality is captured by variants of the Justified Representation axiom~\cite{aziz_2017}, which say that cohesive groups of at least $\ell\cdot\frac{n}{k}$ voters (where $n$ is the total number of voters and $k$ is the desired number of winners) deserve $\ell$ representatives.
The quantity $\frac{n}{k}$ is known as the {\em Hare quota} in the social choice literature. Another---more demanding---choice of quota is the {\em Droop quota}, defined as $\lfloor\frac{n}{k+1}\rfloor+1$. This quota is often used in multiwinner voting with ranked ballots: in algorithms such as Single Transferable Voting, and in proportionality axioms, such as Droop's Proportionality Criterion. A few authors have considered it in the context of approval ballots~\cite{janson_2018,brill_phragmen_2024,aziz_expanding_2020,peters_strategyproofness_2018,peters_core_2025}, but the existing analysis is far from comprehensive. The contribution of our work is a systematic study of JR-style axioms (and voting rules that satisfy them) defined using the Droop quota instead of the Hare quota. For each of the standard JR axioms (namely, JR, PJR, EJR, FPJR, FJR, PJR+ and EJR+), we identify a voting rule that satisfies the Droop version of this axiom. In some cases, it suffices to consider known rules (modifying the corresponding Hare proof, sometimes quite substantially), and in other cases it is necessary to modify the rules from prior work. Each axiom is more difficult to satisfy when defined using the Droop quota, so our results expand the frontier of satisfiable proportionality axioms. We complement our theoretical results with an experimental study, showing that for many probabilistic models of voter approvals, Droop JR/EJR+ are considerably more demanding than standard (Hare) JR/EJR+.
   
\end{abstract}
\newpage

\section{Introduction}
Multiwinner voting with approval ballots is an active subfield of computational social choice~\cite{lackner2023multi}.
An important desideratum in this context is {\em proportionality}, 
i.e., the idea that if a group of voters with shared preferences constitutes a $\sfrac{1}{x}$ fraction of the electorate, they should be able to control a $\sfrac{1}{x}$ fraction of the elected representatives. 

In the context of approval voting, this idea is typically captured by the Justified Representation (JR) axiom~\cite{aziz_2017} and its extensions. This axiom says that in an election with $n$ voters where the goal is to select a size-$k$ committee, if a group of voters $S$ has size at least $\frac{n}{k}$ and all voters in $S$ approve a common candidate, then at least one voter in $S$ must approve some member of the selected committee. Various extensions of this axiom have also been proposed. Typically, they consider larger groups agreeing on multiple candidates,  
and require that such groups get multiple representatives; this includes PJR, EJR, FPJR, FJR, PJR+, and EJR+~\cite{PJR,FPJR,FJR,EJR+}.
Each of these axioms has been shown to be satisfiable, in the sense that every election admits a committee that satisfies it, and for almost all of them (except FJR) there exist polynomial-time computable voting rules that always output committees satisfying these axioms.

A key quantity in the definitions of all these proportionality concepts is the fraction $\frac{n}{k}$, which is known as the {\em Hare quota}; it is named after Thomas Hare, who proposed it in the context of Single Transferable Vote. This quota has a natural interpretation: $\frac{1}{k}$-th of the electorate controls one of the $k$ seats. However, for small values of $k$ the resulting proportionality axiom is very weak: e.g., for $k=1$ the JR axiom is binding only if there is a candidate that is approved by all voters. This weakness of the Hare quota has been recognized by researchers and practitioners alike~\cite{tideman2000better,lundell2007notes}. Thus, a more popular choice of quota in many multiwinner settings is the {\em Droop quota}, named after Henry Richmond Droop, and defined as $\lfloor\frac{n}{k+1}\rfloor+1$. Indeed, this is the quota used in most practical implementation of STV, including national elections in Australia, Ireland and Malta\footnote{\tt https://en.wikipedia.org/wiki/Droop_quota}. For $k=1$, replacing the Hare quota with the Droop quota in the definition of JR results in a meaningful axiom: it rules out outcomes where a majority of voters agree on a candidate, yet none of them approve the election winner. In addition, there cannot exist $k+1$ disjoint groups of voters of size $\lfloor\frac{n}{k+1}\rfloor+1$ each, i.e., satisfying the Droop version of proportionality axioms is not apriori infeasible.

While some work on proportionality in multiwinner voting with approval ballots considers the Droop quota in addition to the Hare quota (see Section~\ref{sec:related} for a literature review), the existing literature has many gaps. This is partly because much of the research on this topic (notably, the important work of Janson~\cite{janson_2018}) precedes the papers that put forward more demanding notions of justified representation---such as FPJR~\cite{FPJR}, FJR~\cite{FJR}, and EJR+~\cite{EJR+}---and sophisticated voting rules that satisfy them, such as, e.g., the Method of Equal Shares~\cite{MES}.
Against this background, the goal of our work is to obtain a comprehensive picture of the Droop proportionality landscape in approval-based multiwinner voting.

\subsection{Our Contribution}
We put forward Droop quota versions of all of the axioms in the JR family, and reproduce many of the key results in the literature for these new axioms. Since the Droop quota is smaller than the Hare quota, we thus show that by slightly modifying existing rules (or, sometimes, keeping the rules unchanged, but tightening the proofs), we can actually satisfy a stronger axiom, guaranteeing representation to smaller groups, and selecting more proportional outcomes.

Our results are summarized in Table~\ref{tbl:results}. Specifically, in Section~\ref{sec:ejr} we focus on Droop EJR/EJR+\cite{aziz_2017,EJR+} (see Section~\ref{sec:prelim} for the definitions of all axioms mentioned below; the voting rules are defined in the sections that prove results for them). We
show that these axioms are satisfied by: (1)~$\varepsilon$-lsPAV~\cite{aziz_2018} with an appropriately chosen value of $\varepsilon$, (2) a natural modification of Greedy Justified Candidate Rule~\cite{EJR+}, and (3) two variants of the Method of Equal Shares~\cite{MES,kraiczy_streamlining_2025} executed with artificially inflated budgets. In Section~\ref{sec:fpjr}, we focus on the recently proposed FPJR axiom~\cite{FPJR}. We show that Droop-FPJR is satisfied by modifications of the Monroe rule and its greedy variant~\cite{monroe} when $k+1$ divides $n$, as well as by all priceable rules~\cite{MES} that select committees of size $k$;
this class of rules includes the Method of Equal Shares (MES)~\cite{MES} completed with SeqPhragm{\'e}n~\cite{brill_phragmen_2024},  
but, contrary to the claim of Kalayc{\i} et al.~\cite{FPJR}, does not include the `vanilla' MES. We therefore provide a separate proof that MES (and the recently proposed Exact Equal Shares~\cite{kraiczy_streamlining_2025}) satisfy (Droop-)FPJR. 
In Section~\ref{sec:fjr} we show that a modification of the Greedy Cohesive Rule~\cite{MES} satisfies Droop FJR. In Section~\ref{sec:experiments} we describe our experiments; the key observation here is that, for the probabilistic models we consider, for many parameter ranges, the Droop versions of the axioms are substantially more demanding than their Hare versions. Section~\ref{sec:concl} concludes and provides future directions. 
Some of the additional experiments are relegated to the appendix, which also includes a number of negative results. E.g., for several rules we show that, even though they satisfy the Hare version of some proportionality axiom, they fail the Droop version of the same axiom; this justifies our proposed modifications of these rules.

\subsection{Related Work}\label{sec:related}
Our paper adapts definitions and theorems from several recent papers on proportional representation in approval-based multiwinner voting~\cite{aziz_2017,brill_phragmen_2024,aziz_2018,MES,EJR+,FPJR,FJR,janson_2018}; we discuss specific papers in relevant sections.

The Droop quota was first proposed by Droop in 1881~\cite{droop}, and multiwinner voting with Droop quota has been extensively studied in the context of apportionment~\cite{pukelsheim2017proportional,brill2024approval} and for ranked ballots---focusing primarily on Single Transferable Vote~\cite{lundell2007notes,tideman2000better}. Aziz and Lee~\cite{aziz_expanding_2020} and Delemazure and Peters~\cite{stv-ties} consider the Droop quota for weak order ballots. In particular, Aziz and Lee put forward the Expanding Approval Rule, and show that it satisfies the Droop Proportionality Criterion \cite{woodall}, a Droop variant of the Proportionality for Solid Coalitions axiom; 
this implies that the approval version of this rule satisfies Droop-PJR. Janson~\cite{janson_2018} studies a variety of properties and voting rules for both approval and ranked ballots, establishing for each property and rule pair the lowest quota (Hare, Droop, or something in between) that the property can be defined with such that the rule satisfies it. Of particular relevance for our work is their proof that PAV satisfies Droop-EJR. Some of the papers on multiwinner voting with approval ballots, while focusing primarily on Hare quota, mention that some of their results extend (or fail to extend) to Droop quota. In particular, Brill et al.~\cite{brill_phragmen_2024} prove that SeqPhragm{\'e}n satisfies Droop-PJR. In contrast, Peters~\cite{peters_strategyproofness_2018} shows that there are stronger impossibility results for satisying strategyproofness and proportionality defined using the Droop quota. Separately, Peters~\cite{peters_core_2025} notes that in the context of core stability there are stronger impossibility results for the variant of the core defined with respect to the Droop quota. 
Masa{\v r}{\'i}k et al.~\cite{masarik_generalised_2024} consider a more expressive model of multiwinner approval voting with constraints; in the absence of constraints the proportionality notion they consider is equivalent to Droop-EJR, and they show that it is satisfied by PAV.
Kehne et al.~\cite{kehne_robust_2025}
consider a variant of Greedy Justified Candidate Rule~\cite{EJR+} that allows groups of size greater than $\frac{\ell n}{k+1}$ to make objections, and claim that it still satisfies EJR+. However, they do not formally define a Droop variant of the EJR+ axiom.

\begin{table}[ht]
\caption{A summary of our main results. Each column refers to the Droop version of the listed axiom. A \cmark\, indicates the rule satisfies the axiom, a \xmark\, indicates the rule does not always satisfy the axiom, and a ? indicates we have no positive or negative results for that pair. \boxmark\textsuperscript{a}\ indicates these only satisfiy PJR when $k+1$ divides $n$, and we do not know if they satisfy PJR+. \boxmark\textsuperscript{b}\ indicates these only satisfiy FPJR when $k+1$ divides $n$. \boxmark\textsuperscript{c}\ indicates that it satisfies PJR/EJR, but we are not sure if it satisfies PJR+/EJR+. \boxmark\textsuperscript{d}\ indicates that it satisfies JR if $k$ divides $n$. The \xmark* for Greedy Monroe indicates that depending on tiebreaking, it may not satisfy JR.} 
\centering
\begin{tabularx}{\textwidth}{@{}l*{5}{lX}@{}}
\toprule
                     & \multicolumn{2}{c}{JR}     & \multicolumn{2}{c}{PJR+}                           & \multicolumn{2}{c}{EJR+}                    & \multicolumn{2}{c}{FPJR}     & \multicolumn{2}{c}{FJR}    \\ \midrule
SeqPhragm{\'e}n  & \cmark  &   & \cmark &\cite{brill_phragmen_2024} & \xmark & \cite{brill_phragmen_2024} & \cmark & Cor. \ref{cor:seq-phragmen-droop-fpjr}  &  \xmark  &    \\
PAV/$\varepsilon$-lsPAV &  \cmark  &    &  \cmark  &   & \cmark & Th. \ref{thm:pav-ejrplus} &  \xmark  & \cite{FPJR} &  \xmark  &    \\
Droop Monroe &  \cmark  &    &  \boxmark\textsuperscript{a}  &  & ? &  & \boxmark\textsuperscript{b} & Th. \ref{thm:droop-monroe-droop-fpjr}, Pr. \ref{prop:droop-monroe-not-droop-pjr} &   ?  &   \\
Droop Greedy Monroe &  \cmark  &    & \boxmark\textsuperscript{a} &  & ? &   & \boxmark\textsuperscript{b} & Th. \ref{thm:droop-monroe-droop-fpjr}, Pr. \ref{prop:droop-monroe-not-droop-pjr} &  ? &   \\
Droop MES/EES &  \cmark  &    &  \cmark &  & \cmark & Th. \ref{thm:droop-mes-droop-ejrplus} & \cmark & Th. \ref{thm:mes-fpjr} & ? &  \\
Droop GJCR  &  \cmark  &    & \cmark &  & \cmark & Th. \ref{thm:droop-gjcr-droop-ejrplus}  & ? &  & ? &  \\
Droop GCR  & \cmark &   & \boxmark\textsuperscript{c} &  & \boxmark\textsuperscript{c} &  & \cmark &  & \cmark & Th. \ref{thm:droop-gcr-droop-fjr} \\ 
Monroe & \boxmark\textsuperscript{d} & Pr. \ref{prop:monroe-droop-jr} & \xmark & Pr. \ref{prop:monroe-not-droop-pjr}  & \xmark &  & \xmark &  & \xmark & \\
Greedy Monroe & \xmark* & Pr. \ref{prop:greedy-monroe-not-droop-jr} & \xmark & Pr. \ref{prop:monroe-not-droop-pjr}& \xmark &  & \xmark &  & \xmark & \\
MES/EES & \xmark & Pr. \ref{prop:mes-not-droop-jr} & \xmark & & \xmark &  & \xmark &  & \xmark &  \\
GJCR & \xmark & Th. \ref{thm:gjcr-gcr-not-droop-jr} & \xmark & & \xmark &  & \xmark & & \xmark & \\
GCR & \xmark & Th. \ref{thm:gjcr-gcr-not-droop-jr} & \xmark & & \xmark &  & \xmark & & \xmark & \\
\bottomrule
\end{tabularx}
\label{tbl:results}
\end{table}

%%%%%%%%%%%%%%%%%%%%%%%%%%%%%%%%%%%%%%%%%%%%%%%%%%%%%%%%%

\section{Preliminaries}\label{sec:prelim}
We first give the formal definition of a multiwinner election with approval ballots. 

\begin{definition}[Multiwinner election with approval ballots]
A \emph{multiwinner election with approval ballots} is a tuple $(C,N,\calA,k)$, where $C$ is the set of {\em candidates}, $N$ is the set of {\em voters}, $\calA = (A_i : i\in N)$ is a list of {\em approval ballots}, where $A_i\subseteq C$ is the set of candidates that voter $i$ approves, and $k$ is the target size of the output committee. An {\em outcome} of $(C,N,\calA,k)$ is a subset of $C$ of size at most $k$.
\end{definition}

We additionally define $n = |N|$, and for each $c\in C$ we write $N_c = \{i\in N : c \in A_i\}$ to denote the set of voters that approve candidate $c$.

\subsection{Voting Rules}
A {\em multiwinner voting rule} is a mapping $\calF$ that, given a tuple $(C, N, \calA, k)$, outputs a non-empty set of size-$k$ subsets of $C$; these are the {\em winning  committees} under $\calF$. We consider several multiwinner voting rules in this paper; to help the reader build intuition, we will now define two of these rules, and postpone the definitions of other rules to the sections where we prove technical results about them.

\begin{definition}[Approval Voting (AV)]\label{def:av}
    The {\em Approval Voting (AV) rule~\cite{approval-handbook}} outputs all size-$k$ committees $W$ 
    that satisfy $|N_w|\ge |N_c|$ for all $w\in W, c\in C\setminus W$. Intuitively, a winning committee contains  $k$ candidates that receive the most votes (up to tie-breaking).
\end{definition}

\begin{definition}[Proportional Approval Voting (PAV)]\label{def:pav}
    The {\em Proportional Approval Voting (PAV) rule~\cite{approval-handbook}} assigns a score to each size-$k$ committee as follows: 
    $\pavscore(W)
    =\sum_{i\in N}\sum_{j=1}^{|W\cap A_i|}\frac{1}{j}$.
    It then outputs all size-$k$ committees with the maximum score.
\end{definition}

%%%%%%%%%%%%%%%%%%%%%%%%%

\subsection{Representation Axioms}
Next, we formulate the representation axioms studied in this paper. We start by defining what it means for a group of voters in an election $(C, N, \calA, k)$ to be cohesive or weakly cohesive; these are the requirements that a group has to meet to deserve representation. Compared to cohesiveness, weak cohesiveness places fewer constraints on the group, and hence leads to stronger proportionality axioms. We present both the Hare and Droop quota versions, which differ only in the size requirement of the cohesive group.

\begin{definition}[Hare/Droop $\ell$-cohesive group]
     For a positive integer $\ell$, we say that a group $S \subseteq N$ is {\em Hare (resp., Droop) $\ell$-cohesive} if $|\bigcap_{i\in S} A_i| \geq \ell$ and $|S| \ge \ell\cdot \frac{n}{k}$ (resp., $|S| > \ell\cdot \frac{n}{k+1}$).
\end{definition}

\begin{definition}[Hare/Droop weakly $(\ell, T)$-cohesive group] For a positive integer $\ell$ and a candidate set $T\subseteq C$, we say that a group $S\subseteq N$ is {\em Hare (resp., Droop) weakly $(\ell,T)$-cohesive} if for each $i\in S$ we have $|A_i\cap T| \geq \ell$ and $|S| \ge |T|\cdot \frac{n}{k}$ (resp., $|S| > |T|\cdot \frac{n}{k+1}$). 
\end{definition}

%EE maybe we'll drop this later, but seems useful to say that
Observe that, if a group of voters $S$ is Hare (resp., Droop) $\ell$-cohesive, then it is Hare (resp., Droop) weakly $(\ell, T)$-cohesive for every set $T$ that is a size-$\ell$ subset of $\bigcap_{i\in S} A_i$. Moreover, if a group is Hare $\ell$-cohesive, it is also Droop $\ell$-cohesive, and if it is Hare weakly $(\ell, T)$-cohesive, it is Droop weakly $(\ell, T)$-cohesive. 

Using these definitions of cohesiveness, we formulate the seven representation axioms that have been studied in the literature, 
both in the standard (Hare) version and in the Droop version. 

For each axiom we define, it is immediate that its Droop version is at least as demanding as its Hare version, i.e., if an election outcome satisfies the Droop version of an axiom, it also satisfies its Hare version;
however, the converse is not true, as shown in \Cref{thm:gjcr-gcr-not-droop-jr} (see also the discussion that precedes this proposition).

%citing the papers that defined the Hare versions of these axioms. 
In line with Aziz et al.~\cite{aziz_2017}, we fix an election $(C, N, \calA, k)$ and define what it means for an outcome $W\subseteq C$ of this election to provide a property X. We say that a rule satisfies X if it always outputs an outcome that provides X.

\begin{definition}[JR \cite{aziz_2017}]
    An outcome $W$ provides {\em Hare (resp., Droop) Justified Representation (JR)} if for every Hare (resp., Droop) $1$-cohesive group $S$, the members of $S$ collectively approve at least one candidate in the outcome, i.e., $\left|\left(\bigcup_{i\in S} A_i\right) \cap W\right| \geq 1$.
\end{definition}

\begin{definition}[PJR \cite{PJR}]
    An outcome $W$ provides {\em Hare (resp., Droop) Proportional Justified Representation (PJR)} if for all $\ell\in[k]$ and for every Hare (resp., Droop) $\ell$-cohesive group $S$, the members of $S$ collectively approve at least $\ell$ candidates in the outcome, i.e., $\left|\left(\bigcup_{i\in S} A_i\right) \cap W\right| \geq \ell$.
\end{definition}

\begin{definition}[FPJR \cite{FPJR}]
    An outcome $W$ provides {\em Hare (resp., Droop) Full Proportional Justified Representation (FPJR)} if for all 
    %EE added a quantifier over T
    $\ell\in[k]$, $T\subseteq C$ and
    for every Hare (resp., Droop) weakly $(\ell, T)$-cohesive group $S$, the members of $S$ collectively approve at least $\ell$ candidates in the outcome, i.e., $\left|\left(\bigcup_{i\in S} A_i\right) \cap W\right| \geq \ell$.
\end{definition}

\begin{definition}[EJR \cite{aziz_2017}]
    An outcome $W$ provides {\em Hare (resp., Droop) Extended Justified Representation (EJR)} if for all $\ell\in[k]$ and for every Hare (resp., Droop) $\ell$-cohesive group $S$, there exists a voter $i\in S$ who approves at least $\ell$ candidates in the outcome, i.e., $|A_i \cap W| \geq \ell$.
\end{definition}

\begin{definition}[FJR \cite{FJR}]
    An outcome $W$ provides {\em Hare (resp., Droop) Full Justified Representation (FJR)} if for all $\ell\in[k]$ and for every Hare (resp., Droop) weakly $(\ell,T)$-cohesive group $S$, there exists a voter $i\in S$ who approves at least $\ell$ candidates in the outcome, i.e., $|A_i \cap W| \geq \ell$.
\end{definition}

The JR axiom can be viewed as a special case of PJR and EJR when $\ell=1$, and is fairly easy to satisfy (while AV fails it, it is satisfied by PAV and many other rules); thus, in what follows we will focus on the other axioms.
PJR is the weakest among the other four axioms. EJR and FPJR are both strengthenings of PJR, with EJR requiring a stronger guarantee for groups that are cohesive, and FPJR allowing more groups (namely weakly cohesive ones) to demand representation. EJR and FPJR are known to be incomparable for the Hare quota~\cite{FPJR}, and we will conclude the same for the Droop versions (Corollary~\ref{cor:ejr-fpjr-incomparable}). FJR strengthens EJR and FPJR by combining the requirements of each.

 The next two axioms, PJR+ and EJR+, are strengthenings of PJR and EJR. These axioms relax the requirement for a group to be $\ell$-cohesive---in a different way than weakly cohesive groups. Namely, instead of requiring a group of voters to jointly approve $\ell$ candidates---as in the case of cohesive groups---it only requires them to jointly approve one candidate that is not in the winning committee.

\begin{definition}[PJR+ \cite{EJR+}]
    \label{def:pjrplus}
    An outcome $W$ provides {\em Hare (resp., Droop) PJR+} if for all $\ell \in [k]$, every group of voters $S$ with size $|S|\ge\ell\cdot\frac{n}{k}$ (resp, $|S|>\ell\cdot\frac{n}{k+1}$) such that $ \bigcap_{i\in S} A_i\setminus W\neq\varnothing$, the members of $S$ collectively approve at least $\ell$ candidates in the outcome, i.e., $\left|\left(\bigcup_{i\in S} A_i\right) \cap W\right| \geq \ell$.
\end{definition}

\begin{definition}[EJR+ \cite{EJR+}]
    \label{def:ejrplus}
    An outcome $W$ provides {\em Hare (resp., Droop) EJR+} if for all $\ell \in [k]$, every group of voters $S$ with size $|S|\ge\ell\cdot\frac{n}{k}$ (resp., $|S|>\ell\cdot\frac{n}{k+1}$) such that $ \bigcap_{i\in S} A_i\setminus W\neq\varnothing$
    there exists some voter $i\in S$ who approves at least $\ell$ candidates in the outcome, i.e., $|A_i \cap W| \geq \ell$.
\end{definition}

%Brill and Peters~\cite{EJR+} show that in the approval setting, the Hare version of PJR+ is equivalent to IPSC, a proportionality axiom for participatory budgeting with ordinal preferences, defined by Aziz and Lee~\cite{IPSC}. Accordingly, Droop-PJR+ can be viewed as an approval version of IPSC with the Droop quota; this version of IPSC was defined and studied by Aziz et al.~\cite{Droop-IPSC}. Finally, 
Brill and Peters~\cite{EJR+} show that for the Hare quota, EJR+ is incomparable to FJR. In this paper we show the same is true for the Droop versions (Corollary~\ref{cor:incomparability-ejrplus-fjr}).

In what follows, for consistency with prior work, we will often omit `Hare' from the name of an axiom, i.e., we will write `EJR' instead of `Hare-EJR'.

The following lemma will be useful in our analysis.

\begin{lemma}
    \label{lem:size-of-cohesive-group}
    For all $z, \ell, n, k\in\mathbb N$ the inequality $z > \frac{\ell n}{k+1}$ implies
    $z \geq \frac{\ell n + 1}{k+1}$.
\end{lemma}

\begin{proof}
    Let $s=z(k+1)$; since $z\in\mathbb N$ so is $s$. Then  $z - \frac{\ell n}{k+1} = \frac{s-\ell n}{k+1}$. Now, $z>\frac{\ell n}{k+1}$ implies $s-\ell n > 0$; as $s$ and $\ell n$ are both integers, we obtain $s- \ell n \geq 1$. Hence $z - \frac{\ell n}{k+1} \geq \frac{1}{k+1}$, and the claim follows.
    \qed
\end{proof}
%%%%%%%%%%%%%%%%%%%%%%%%%%%%%%%%%%%%%%%%%%%%%%%%%%%%%%%%%%%%%%%%%%%%%%%%%%

\section{Extended Justified Representation(+)}\label{sec:ejr}
We will now consider several voting rules that are known to satisfy EJR+, and prove that they satisfy---or can be modified to satisfy---Droop-EJR+.

\subsection{Local Search PAV}
The first voting rule that was shown to satisfy EJR was the PAV rule~\cite{aziz_2017}. Subsequently, Brill and Peters~\cite{EJR+} showed that it also satisfies EJR+ and Janson~\cite{janson_2018} showed that it satisfies Droop-EJR. However, computing the winning committees under PAV is NP-hard~\cite{PAV-hard}, making this rule unsuitable for practical use. To address this, Aziz et al.~\cite{aziz_2018} proposed a bounded local search variant of PAV, which we will refer to as $\varepsilon$-lsPAV. For a suitable choice of $\varepsilon$, this rule is polynomial-time
computable and satisfies EJR (the proof of Aziz et al.~\cite{aziz_2018} can also be used to show that it satisfies EJR+, but to the best of our knowledge this observation has not been made in the literature). However, it was not known whether this rule satisfies Droop-EJR. We will now close this gap, showing that we can choose $\varepsilon$ so that $\varepsilon$-lsPAV satisfies Droop-EJR+ and is polynomial-time computable. We start by giving a formal definition of this rule.

\begin{definition}[$\boldsymbol{\varepsilon}$-lsPAV~\cite{aziz_2018}]
    \label{def:ls-pav}
    The {\em $\varepsilon$-bounded local search PAV rule ($\varepsilon$-lsPAV)} starts with an arbitrary size-$k$ committee $W$ and proceeds in rounds. At each round it checks if there is a pair of candidates $(w, c)\in (W,C\setminus W)$ such that $\pavscore(W\cup\{c\}\setminus\{w\})\ge \pavscore(W)+\varepsilon$;
    if some such pair exists, it sets $W:=W\cup\{c\}\setminus\{w\}$.
    When no such pair can be found, the rule returns $W$.
\end{definition}

Aziz et al.~\cite{aziz_2018} show that $\frac{n}{k^2}$-lsPAV runs in polynomial time and satisfies EJR; we will now extend their result to
Droop-EJR+. Note that while the general proof strategy is similar in spirit to that of Aziz et al.~\cite{aziz_2018}, the particulars of the proofs are different: our proof is forced to take a more careful approach, because of the more stringent Droop quota.

\begin{theorem}
\label{thm:pav-ejrplus}
    $\frac{1}{k^2}$-lsPAV is polynomial-time computable and  satisfies Droop-EJR+.
\end{theorem}
\begin{proof}
    The argument that $\frac{1}{k^2}$-lsPAV runs in polynomial time is the same as in the work of Aziz et al.~\cite{aziz_2018}: each swap increases the PAV score by at least $\frac{1}{k^2}$, and the maximum possible PAV score is $n(1+\frac12+\dots+\frac{1}{k})$, so the number of iterations is $O(nk^2\log k)$, and each iteration runs in time $O(k|C|\cdot nk)$.

    To argue that $\frac{1}{k^2}$-lsPAV satisfies Droop-EJR+, we will show that if a size-$k$ committee $W$ fails to provide Droop-EJR+ then there exists a pair of candidates $(w, c)\in (W, C\setminus W)$ such that replacing $w$ with $c$ increases the PAV score by at least $\frac{1}{k^2}$; hence, when the algorithm terminates, $W$ provides Droop-EJR+.
 
    Fix a committee $W$ that fails Droop-EJR+, as witnessed by $\ell\in [k]$, a group $S$ with $|S|>\frac{\ell n}{k+1}$ and $|A_i \cap W| \leq \ell - 1$ for all $i\in S$, and a candidate $c\in \bigcap_{i\in S} A_i\setminus W$. 
    Let $A_{S} = \bigcap_{i\in S} A_i$,  and for every $w\in W$ let $m(w)=\pavscore(W) - \pavscore(W\setminus\{w\})$ be the marginal contribution of $w$. Also, let $W_i=A_i\cap W$ for each $i\in N$.
    
    Note that $|A_S \cap W|\le \ell-1$ and hence $W\setminus A_S\neq\varnothing$. Suppose we replace some $w \in W\setminus A_{S}$ with $c$; let $\Delta(w,c)$ denote the resulting change in PAV score. Since each $i\in S$ approves $c$, we have 
\begin{align*}
        \Delta(w,c) \geq \sum_{i\in S} \frac{1}{|W_i\setminus\{w\}| + 1} - m(w)
        \geq \!\!\sum_{i\in S : w\notin A_i}\frac{1}{|W_i| + 1} + \!\!\sum_{i\in S : w\in A_i}\frac{1}{|W_i|} - m(w).
    \end{align*}
     Taking a sum over all candidates in $W\setminus A_{S}$, we get
\begin{align*}       \sum_{w\in W\setminus A_{S}} \Delta(w,c) &\geq \!\sum_{w\in W\setminus A_{S}}\!\!\left(\sum_{i\in S : w\notin A_i}\frac{1}{|W_i| + 1} + \sum_{i\in S : w\in A_i}\frac{1}{|W_i|} \right) - \sum_{w\in W\setminus A_{S}} m(w)\\
        &= \sum_{i\in S}\left(\sum_{\substack{w\in W\setminus A_{S} \\ w\notin A_i}} \frac{1}{|W_i| + 1} + \sum_{\substack{w\in W\setminus A_{S} \\ w \in A_i}} \frac{1}{|W_i|} \right) - \sum_{w\in W\setminus A_{S}} m(w)\\
        &= \!\!\sum_{\substack{i\in S\\ |W_i| \geq 1}}\left(\frac{|(W\setminus A_{S}) \setminus A_i|}{|W_i| + 1} + \frac{|(W\setminus A_{S}) \cap A_i|}{|W_i|} \right) 
        +\sum_{\substack{i\in S\\ W_i = \varnothing}}\frac{|(W\setminus A_{S}) \setminus A_i|}{|W_i| + 1}- \sum_{w\in W\setminus A_{S}} m(w)\\
        &= \!\!\sum_{\substack{i\in S\\ |W_i| \geq 1}}\!\!\left(\frac{k-|W_i|}{|W_i| + 1} + \frac{|W_i| - |W\cap A_{S}|}{|W_i|} \right)
        +\!\!\sum_{\substack{i\in S\\ W_i = \varnothing}}k - \!\!\!\sum_{w\in W\setminus A_{S}} m(w).
    \end{align*}
    Note that each voter $i\in N$ with $W_i = A_i\cap W\neq\varnothing$ contributes exactly 1 to $\sum_{w\in W} m(w)$: if $|W_i|=j$, then each candidate in $W_i$ provides a marginal contribution of $\frac{1}{j}$ to $i$'s `PAV utility' . Therefore we obtain
  \begin{align*}
     \sum_{w\in W\setminus A_{S}} m(w) &= \sum_{w\in W} m(w) - \sum_{w\in W\cap A_{S}} m(w)
        = \sum_{\substack{i\in N\\ |W_i| \geq 1}}1 - \sum_{w\in W\cap A_{S}} m(w)\\
        &\leq n - \sum_{\substack{i\in S\\W_i = \varnothing}} 1 - \sum_{w\in W\cap A_{S}} m(w)
        \leq n - \sum_{\substack{i\in S\\W_i = \varnothing}} 1 - \sum_{\substack{i\in S\\ |W_i| \geq 1}}\frac{|W\cap A_{S}|}{|W_i|}.
    \end{align*}
    Combining these inequalities, we get
    \begin{align*}
        \sum_{w\in W\setminus A_{S}} \Delta(w,c) &\geq \sum_{\substack{i\in S\\ |W_i| \geq 1}}\left(\frac{k-|W_i|}{|W_i| + 1} + \frac{|W_i| - |W\cap A_{S}|}{|W_i|} \right)+\sum_{\substack{i\in S\\ W_i = \varnothing}}k
         - n + \sum_{\substack{i\in S\\W_i = \varnothing}} 1 + \sum_{\substack{i\in S\\ |W_i| \geq 1}}\frac{|W\cap A_{S}|}{|W_i|}\\
        &= \sum_{\substack{i\in S\\ |W_i| \geq 1}} \left(\frac{k-|W_i|}{|W_i|+1} + \frac{|W_i|}{|W_i|}\right) + \sum_{\substack{i\in S\\ W_i = \varnothing}} (k+1) - n\\
        &= \sum_{\substack{i\in S\\ |W_i| \geq 1}} \frac{k+1}{|W_i| + 1}+ \sum_{\substack{i\in S\\ W_i = \varnothing}} (k+1) -n
        \geq \sum_{\substack{i\in S\\ |W_i| \geq 1}} \frac{k+1}{(\ell -1) + 1}+ \sum_{\substack{i\in S\\ W_i = \varnothing}} (k+1) -n\\
        &\geq \sum_{i\in S} \frac{k+1}{\ell} -n= |S|\cdot \frac{k+1}{\ell} -n\geq \frac{\ell n + 1}{k+1} \cdot \frac{k+1}{\ell} -n= \frac{1}{\ell}, 
    \end{align*}
    where we use Lemma~\ref{lem:size-of-cohesive-group} to lower-bound $|S|$.
    
    Hence, by the pigeonhole principle there is some candidate $w \in W\setminus A_{S}$ for which $\Delta(w,c) \geq \frac{1}{\ell k} \geq \frac{1}{k^2}$.
    \qed
\end{proof}

Brill and Peters~\cite{EJR+} prove  that for the Hare quota, EJR+ is incomparable to FJR; we will now prove an analog of their result for the Droop quota.
\begin{corollary}
\label{cor:incomparability-ejrplus-fjr}
    Droop-EJR+ and Droop-FJR are incomparable.
\end{corollary}

\begin{proof}
    PAV satisfies Droop-EJR+, but does not satisfy FPJR~\cite{FPJR}, which means that it also does not satisfy Droop-FJR. Now, consider the following example due to Brill and Peters~\cite{EJR+}: there are 3 candidates $c_1,c_2,c_3$, two voters with approval ballots $A_1=\{c_1, c_2\}$ and $A_2=\{c_1, c_3\}$ and $k=2$. Then the outcome $\{c_2,c_3\}$ satisfies Droop-FJR but does not satisfy Droop-EJR+.
    \qed
\end{proof}

\subsection{Greedy Justified Candidate Rule}
Brill and Peters~\cite{EJR+} propose another polynomial-time computable rule---Greedy Justified Candidate Rule (GJCR)---that is explicitly designed to satisfy Hare-EJR+. This rule operates by finding groups of size $\ell\cdot\frac{n}{k}$ that are unsatisfied with the current outcome, i.e., it is defined with the Hare quota in mind. Therefore, it is not surprising that it does not satisfy Droop-EJR+; in \Cref{thm:gjcr-gcr-not-droop-jr} (Appendix~\ref{app:neg}) we show that this is indeed the case. However, we will now show that, by replacing the condition $|S|\ge\ell\cdot\frac{n}{k}$ with $|S|>\ell\cdot\frac{n}{k+1}$ in the definition of this rule, we obtain a rule that satisfies Droop-EJR+.
We start by defining both variants of this rule. 

\begin{definition}[Greedy Justified Candidate Rule (GJCR)~\cite{EJR+}]\label{def:gjcr}
    The {\em Hare (resp., Droop) Greedy Justified Candidate Rule (GJCR)} starts by setting $W=\varnothing$ and $\ell = k$, and proceeds iteratively. In each round, it
    checks whether there is a candidate $c\in C\setminus W$ such that there is a group of voters $S\subseteq N_c$ with
    $|S|\ge \ell\cdot\frac{n}{k}$ (resp., $|S|>\ell\cdot\frac{n}{k+1}$) such that each voter in $S$ approves at most $\ell-1$ candidates in $W$. If yes, it adds some such candidate $c$ to $W$; otherwise, it decrements $\ell$ by 1. If $\ell=0$, it adds an arbitrary set of $k-|W|$ candidates to $W$ and outputs the resulting committee.
\end{definition}
 We now adapt the proof of Brill and Peters~\cite{EJR+} that GJCR satisfies Hare-EJR+ to show that Droop GJCR satisfies Droop-EJR+.

\begin{theorem}
    \label{thm:droop-gjcr-droop-ejrplus}
    Droop GJCR selects a committee of size $k$ and satisfies Droop-EJR+.
\end{theorem}

\begin{proof}
    Let $W$ be an output of Droop GJCR. It is immediate that $W$ satisfies Droop-EJR+. Indeed, if there is an $\ell'\in \mathbb N$, a candidate $c\in C\setminus W$, and a group of voters $S$ with $|S| > \frac{\ell' n}{k+1}$ such that $c\in A_i$, $|A_i\cap W|<\ell'$ for all $i\in S$, then
    $c$ would have been selected by the algorithm when $\ell=\ell'$, a contradiction.

    Now we show that $|W| = k$. To this end, we set up a pricing scheme with total budget of $k+1 - \frac{1}{n}$  and a per-candidate cost of 1. The existence of this scheme proves that the rule selects at most $k$ candidates. We start by giving each voter a budget of $\frac{k+1}{n} - \frac{1}{n^2}$. If a candidate $c$ is selected because of a voter set $N' = \{i\in N_c \,:\, |A_i \cap W| < \ell\}$ with $|N'| > \frac{\ell n }{k+1}$ (and hence by Lemma~\ref{lem:size-of-cohesive-group} $|N'|\ge \frac{\ell n+1}{k+1}$), we split the (unit) cost of $c$ equally among the voters in $N'$, so that each voter in $N'$ pays $\frac{1}{|N'|} \le\frac{k+1}{\ell n+1}$ for $c$. Consider a voter $i\in N'$. Note that up to this point $i$ only spent their budget on candidates whose costs were shared by groups of size greater than $\frac{\ell n }{k+1}$, and $|A_i\cap W|<\ell$. Hence, before $c$ is selected, $i$'s total spending is at most $(\ell-1)\cdot\frac{k+1}{\ell n+1}$. 
    Since $k\ge \ell$, we have
    $$
    \frac{k+1}{n}\cdot \left(1-\frac{\ell n}{\ell n+1}\right) = \frac{k+1}{n}\cdot\frac{1}{\ell n+1}\ge \frac{1}{n}\cdot\frac{k+1}{\ell n+n}\ge \frac{1}{n^2}. 
    $$
    Therefore, we can bound the total spending of voter $i$ after purchasing $c$ as 
    $$
    \ell\cdot\frac{k+1}{\ell n+1} = \ell\cdot\frac{k+1}{\ell n}\cdot\frac{\ell n}{\ell n+1}=\frac{k+1}{n}\cdot\frac{\ell n}{\ell n + 1}\le \frac{k+1}{n}-\frac{1}{n^2}.
    $$
    As this is true for every voter at every point in the execution of the algorithm, no voter ever 
    overspends their budget. Since the total budget of the voters is less than $k+1$, we have that at most $k$ candidates are purchased.
    \qed
\end{proof}

\subsection{Equal Shares Rules}

So far in the section, we considered PAV, $\varepsilon$-lsPAV and GJCR. The PAV rule is not defined in terms of quotas, in the sense that the quantity $\frac{n}{k}$ (or $\frac{n}{k+1}$) does not appear in the definition of the rule. Accordingly, we did not have to modify the rule for it to satisfy Droop-EJR+ (though for the local search version we did have to use a smaller value of $\varepsilon$, compared to the one used to make this rule satisfy Hare-EJR). In contrast, GJCR is defined in terms of a quota, so, to create a version of GCJR that satisfies Droop-EJR+, we had to tweak the rule itself. 

The next rule we consider is the Method of Equal Shares (MES)~\cite{MES}, together with its recently proposed simplification, Exact Equal Shares (EES)~\cite{kraiczy_streamlining_2025}. 

\begin{definition}[Method of Equal Shares (MES)~\cite{MES}]
The rule proceeds in a sequential manner, starting with $W=\varnothing$. Each candidate has a cost of $\frac{n}{k}$, and initially each voter's budget $b_i$ is $1$. At each iteration, the rule computes the {\em affordability threshold} $q(c)$ of each candidate $c\in C\setminus W$ as the smallest value $q$ that satisfies $\sum_{i\in N_c}\min\{b_i, q\}=\frac{n}{k}$ (i.e., the voters who `purchase' $c$ have to share its cost equally, except that if a voter would run out of money by doing so, they can contribute their entire remaining budget instead). It then selects a candidate $c\in\arg\min q(c)$, adds it to $W$ and updates the budgets of voters in $N_c$ as $b_i:=\max\{b_i-q, 0\}$.
The algorithm terminates and returns $W$ when no candidate in $C\setminus W$ has a bounded affordability threshold; importantly, it may happen that $|W|<k$.
%The rule builds the committee sequentially, with voters purchasing a candidate at each step. The cost of a candidate is split amongst its approvers as evenly as possible. If the cost of the candidate cannot be split exactly equally because some voters do not have enough budget remaining, then other approving voters can pay more to compensate. We refer to the maximum cost that any voter pays for a candidate as the priceability of that candidate. At each step the rule purchases the candidate that has the lowest priceability. If there is no candidate that can be afforded by its approvers then the rule terminates and outputs the committee it has built up so far.\footnote{It is possible for MES to return a committee of size less than $k$.}
\end{definition}

Exact Equal Shares (EES)~\cite{kraiczy_streamlining_2025} is a variant of MES where the cost of a candidate $c$ must be split exactly equally amongst all voters who pay for it, 
i.e., the affordability threshold of $c$ is defined as $\min\{\frac{n}{k\cdot |S|}: S\subseteq N_c, b_i\ge \frac{n}{k\cdot|S|}\}$ (this quantity can be computed efficiently by a greedy algorithm). Our proof techniques are general enough to apply to both MES and EES. 

 The quantity $\frac{n}{k}$ appears in the definition of this rule, so it is perhaps  not surprising that MES fails Droop-EJR+ (see \Cref{prop:mes-not-droop-jr} in Appendix~\ref{app:neg}). A natural approach to address this would be to set the candidate costs to $\frac{n}{k+1}+\varepsilon$ for a carefully chosen value of $\varepsilon$: $\varepsilon$ should be positive, so that the voters cannot afford more than $k$ candidates, but small enough that a group of size greater than $\ell\cdot\frac{n}{k+1}$ can afford $\ell$ candidates. It turns out that this indeed results in a rule that satisfies Droop-EJR+. For presentation purposes, instead of scaling down the candidate costs, we will scale up the voters' budgets. 

In more detail, we will run MES with a \emph{virtual budget}, allocating each voter $i$ a budget of $b > 1$. We note that executing MES with a virtual budget is a common technique used to force this rule to fill as many of the $k$ seats as possible: indeed, when run with $b=1$, MES frequently selects much fewer than $k$ candidates (see the discussion in the work of Kraiczy et al.~\cite{kraiczy_adaptive_2024,kraiczy_streamlining_2025}).  
It turns out that, by setting $b=\frac{(k+1)n}{kn+1}=1+\frac{n-1}{kn+1}$, 
we can ensure that  
 MES with budget $b$ selects at most $k$ candidates and satisfies Droop-EJR+. We will refer to the variants of MES/EES that use this value of $b$ as {\em Droop MES/EES}.

\begin{theorem}
    \label{thm:droop-mes-droop-ejrplus}
     When run at a virtual budget of $\frac{(k+1)n}{kn+1}$, MES/EES select at most $k$ candidates and satisfy Droop-EJR+.
    %$\frac{k+1}{k} - \varepsilon$, for $\varepsilon \leq \frac{k+1}{k}\cdot\frac{1}{kn+1}$.
\end{theorem}

\begin{proof}
    First note that $\frac{n}{kn+1}<\frac{1}{k}$ and hence $n\cdot \frac{(k+1)n}{kn+1}< (k+1)\cdot\frac{n}{k}$. Therefore, collectively the voters can purchase at most $k$ candidates.
        
    To prove that MES/EES with this virtual budget satisfy Droop EJR+, we assume for contradiction 
    that on some election these rules output a committee $W$ that does not provide Droop-EJR+. That is,  there exists a group of voters $S$ with $|S| > \ell\cdot\frac{n}{k+1}$ and $|W \cap A_i| \leq \ell -1$ for all $i\in S$, and a candidate $c \in \bigcap_{i\in S} A_i\setminus W$. 
    
    Suppose first that $\ell=1$. Then each voter in $S$ approves no candidates in $W$ and therefore still has her original budget. By Lemma~\ref{lem:size-of-cohesive-group} we have $|S|\ge \frac{n+1}{k+1}$, so voters in $S$ collectively have at least 
    $$
    \frac{n+1}{k+1}\cdot\frac{(k+1)n}{kn+1}= \frac{n(n+1)}{nk+1}\ge \frac{n(n+1)}{nk+k}=\frac{n}{k}
    $$ 
    dollars and can afford to share the cost of $c$ equally, a contradiction. 

    Thus, from now on we will assume $\ell\ge 2$.
    We claim that the budget of some voter $i'\in S$ is less than $\frac{n}{k|S|}$. Indeed, if not, then the rule would purchase $c$, splitting its cost $\frac{n}{k}$ among the voters in $S$.
    Since $i'$ paid for at most $\ell-1$ candidates in $W$,  there exists a candidate $c'\in W$ such that $i'$ spent more than 
    \begin{equation}
    \label{eq:MES spent more than}
        q = \frac{1}{\ell - 1}\cdot\left(\frac{(k+1)n}{kn+1} - \frac{n}{k|S|}\right)
    \end{equation}
    dollars on $c'$. Consider the first time the algorithm bought a candidate $c^*$ whose affordability threshold was  greater than $q$. We claim that at this time, each voter $i\in S$ had at least $\frac{n}{k|S|}$ dollars left. Indeed, we have
    \begin{align*}
     \frac{(k+1)n}{kn+1} - (\ell -1)\cdot\frac{1}{\ell -1}\left(\frac{(k+1)n}{kn+1} - \frac{n}{k|S|}\right)
        = \frac{n}{k|S|}.
    \end{align*}

    Thus at this point in time, the members of $S$ could buy $c$ at affordability threshold of at most $\frac{n}{k|S|}$. To obtain a contradiction, we will show that $\frac{n}{k|S|}\le q$. Indeed, by Lemma~\ref{lem:size-of-cohesive-group} we have $|S|\ge \frac{\ell n+1}{k+1}$ and hence with $\ell\le k$, this implies
    $
    \frac{n\ell}{k|S|}\le \frac{n(k+1)\ell}{k(\ell n+1)} \leq \frac{n(k+1)\ell}{k\ell n + \ell} = \frac{n(k+1)}{kn+1}
    $.
    Therefore, 
    $$
    q(\ell-1) - \frac{n(\ell-1)}{k|S|} = 
    \frac{(k+1)n}{kn+1}-\frac{n}{k|S|} - \frac{n(\ell-1)}{k|S|} = 
    \frac{(k+1)n}{kn+1} - \frac{n\ell}{k|S|}\ge \frac{(k+1)n}{kn+1} - \frac{n(k+1)}{kn+1} = 0;
    $$
    dividing both sides by $\ell-1$, we conclude that 
$q\ge \frac{n}{k|S|}$.
    But this is a contradiction, since the algorithm buys the candidate with the lowest affordability threshold, and we know that the affordability threshold of $c^*$ is strictly larger than $q$, i.e., higher than that of $c$. We conclude that MES/EES satisfy Droop-EJR+.
    \qed
\end{proof}
%%%%%%%%%%%%%%%%%%%%%%%%%%%%%%%%%%%%%%%%%%%%%%%%%%%%%%%

\section{Full Proportional Justified Representation}\label{sec:fpjr}
Next, we consider Proportional Justified Representation~\cite{PJR} and the recently introduced axiom of Full Proportional Justified Representation~\cite{FPJR}.

\subsection{Monroe Rules}

The first positive result for PJR was established by S\'anchez-Fern\'andez et al.~\cite{PJR}, who showed that the Monroe rule and its greedy variant satisfy PJR if the target committee size $k$ divides the number of voters $n$. Subsequently, Kalayc{\i} et al.~\cite{FPJR} extended this result to FPJR.
We will now define the Monroe rule and the Greedy Monroe rule. As these rules 
are defined in terms of quotas, we give both the standard definition (corresponding to the Hare quota) and a modified definition (corresponding to the Droop quota). We then show that the Droop variants of both rules provide Droop-FPJR as long as $k+1$ divides $n$.

\begin{definition}[Monroe Rule~\cite{monroe}]\label{def:monroe}
    Fix a dummy candidate $d\not\in C$.
    A {\em Hare valid assignment} is a pair $(W, \pi)$, where $W$ is a size-$k$ subset of $C$ and $\pi: N\rightarrow W$ is a mapping that satisfies $\lfloor \frac{n}{k}\rfloor \leq |\pi^{-1}(c)| \leq \lceil \frac{n}{k} \rceil$ for all $c\in W$.
    A {\em Droop valid assignment} is a pair $(W, \pi)$, where $W$ is a size-$k$ subset of $C$ and $\pi: N\rightarrow W\cup\{d\}$ is a mapping that satisfies (1)
        $\lfloor \frac{n}{k+1}\rfloor \leq |\pi^{-1}(c)| \leq \lceil \frac{n}{k+1} \rceil$ for all $c\in W$ and (2) $|\pi^{-1} (d)| = \lfloor \frac{n}{k+1} \rfloor$.
    The {\em Monroe score} of a Hare/Droop valid assignment $(W, \pi)$ is computed as 
    $\sum_{i\in N} \bbbone{\{\pi(i) \in A_i\}}$, and the {\em Hare (resp., Droop) Monroe score} of a committee $W$ is the maximum Monroe score of a Hare (resp., Monroe) valid assignment $(W, \pi)$, computed over all possible choices of $\pi:N\to W$ (resp., $\pi: N\rightarrow W\cup\{d\}$).
    The {\em Hare (resp., Droop) Monroe rule} outputs the set of all size-$k$ committees
    that maximize the Hare (resp., Droop) Monroe score.
\end{definition}

\begin{definition}[Greedy Monroe Rule~\cite{PJR}]
    The Hare (resp., Droop) Greedy Monroe rule starts with $W=\varnothing$ and all voters marked as active. It proceeds in $k$ rounds. In round $t$, it does the following:
    \begin{enumerate}
        \item Finds a candidate $c\in C\setminus W$ that receives the maximum number of approvals from the active voters.
        \item 
        Assigns roughly $\frac{n}{k}$ (resp, $\frac{n}{k+1}$) active voters to $c$. 
        Specifically, Hare Greedy Monroe assigns $\lceil \frac{n}{k}\rceil$ voters
        to $c$ if $t\le n-k\lfloor\frac{n}{k}\rfloor$ and $\lfloor\frac{n}{k}\rfloor$ voters
        otherwise. Droop Greedy Monroe assigns $\lceil \frac{n}{k+1}\rceil$ voters to $c$
        if $t\le n-(k+1)\cdot\lfloor\frac{n}{k+1}\rfloor$ and $\lfloor\frac{n}{k+1}\rfloor$ voters
        otherwise.
        %the number of active voters is divisible by $k+1-|W|$ and $\lceil \frac{n}{k+1}\rceil$ voters otherwise. 
        As many of the assigned voters as possible should be selected from $N_c$, and the rest are arbitrarily chosen;
        \item Adds $c$ to $W$, and marks all voters assigned to $c$ as inactive.
    \end{enumerate}
\end{definition}
Note that Greedy Monroe implicitly constructs a valid assignment of voters to candidates (in case of Droop Greedy Monroe we can think of voters that remain active at the end as being assigned to a dummy candidate $d\not\in C$), i.e., we can speak of an assignment $(W, \pi)$ associated with the winning committee $W$. 

We begin by proving a useful fact about the Droop Monroe rule. 

\begin{lemma}
    \label{lem:monroe-unsatisfied-voters}
     Consider an election $(N,C,\calA, k)$ such that $k+1$ divides $n$, 
     a valid Droop assignment $(W, \pi)$ with the maximum Monroe score, and a group of voters $S\subseteq N$ with size $|S| > \frac{n}{k+1}$ such that for all voters $i\in S$, $\pi(i) \notin A_i$. Then $\bigcap_{i\in S} A_i \subseteq W$.
\end{lemma}

\begin{proof}
    Assume for contradiction that there is a group of voters $S$ with $|S| > \frac{n}{k+1}$, $\pi(i) \notin A_i$ for all $i\in S$, and a candidate
    $c' \in (\bigcap_{i\in S} A_i)\setminus W$. Since $|S|>\frac{n}{k+1}$ and the Droop Monroe rule assigns exactly $\frac{n}{k+1}$ voters to each candidate  in $W\cup\{d\}$, we have $\pi(i^*)\neq d$ for some $i^*\in S$; let $c=\pi(i^*)$.
    
    We will now argue that replacing $c$ with $c'$ in $W$, i.e., setting $W'=W\cup\{c'\}\setminus\{c\}$,  increases the Droop Monroe score. To this end, we construct the mapping $\pi': N\to W'$ as follows.
    Let $N'=\{i\in N\setminus S : \pi(i)=c\}$; since there are exactly $\frac{n}{k+1}$ voters with $\pi(i)=c$, and one of these voters is $i^*$, who is not in $N'$, we have $|N'|\le \frac{n}{k+1} -1$.
    As $|S|>\frac{n}{k+1}$, there is a proper subset $S'\subset S$ with $|S'|=|N'|$ such that $\pi(i)\neq c$ for all $i\in S'$.
    Let $\sigma:N'\to S'$ be a bijection between $N'$ and $S'$, and set 
    $\pi'(i):=\pi(\sigma(i))$, 
    $\pi'(\sigma(i)):=c'$
    for each $i\in N'$.
    Further, for each $i\in S$ with $\pi(i)=c$ set $\pi'(i):=c'$. For all voters in $N\setminus N'$ and all voters in $S\setminus S'$ with $\pi(i)\neq c$ set $\pi'(i):=\pi(i)$.
    Note that $(W', \pi')$ is a valid assignment with $\frac{n}{k+1}$ voters in $S$ assigned to $c'$.

    We claim that the Monroe score of $(W', \pi')$ is strictly higher than that of $(W, \pi)$. Indeed, compared to $(W, \pi)$, the assignment $(W', \pi')$ benefits the $\frac{n}{k+1}$ voters in $S$ that are now assigned to $c'$ and harms at most $|N'|\le \frac{n}{k+1}-1$ voters in $N'$; all other voters are assigned to the same candidates under $\pi$ and $\pi'$. This is a contradiction with our choice of $(W, \pi)$.
    %
    %modify the assignment function as follows. First assign all voters that were assigned to $c$ to $c'$. Then take all the voters not in $S$ that are assigned to $c'$ and swap them with some voters in $S$. There can be at most $\frac{n}{k+1} - 1$ such voters since we know at least one member of $S$ was assigned to $c$. Thus even if none of these voters like their new candidate, the Monroe objective goes down by at most $\frac{n}{k+1} -1$. On the other hand, we now have that all the voters assigned to $c'$ are from $S$. This increases the objective by $\frac{n}{k+1}$ since all of those voters go from not approving the candidate they are assigned to, to approving it. Thus the Monroe objective strictly increases. But this is a contradiction, since the Droop Monroe rule outputs a $W$ and $\pi$ that maximize the objective. We conclude that $\bigcap_{i\in S} A_i \subseteq W$.
    \qed
\end{proof}

\begin{theorem}
    \label{thm:droop-monroe-droop-fpjr}
    The Droop Monroe rule and Droop Greedy Monroe rule satisfy Droop-FPJR if $k+1$ divides $n$.
\end{theorem}

\begin{proof}
    The proofs for the two rules are very similar, so we combine them, noting explicitly when they diverge. Assume for the sake of contradiction that there is an election $(C, N, \calA, k)$ on which the Droop Greedy Monroe rule (resp., Droop Monroe rule) outputs a committee $W$ associated with assignment $\pi$ that does not provide Droop-FPJR. Then there is a Droop weakly $(\ell,T)$-cohesive group $S$ such that $|W\cap \bigcup_{i\in S} A_i| < \ell$. 
    %EE we have A_S in another proof denoting intersection of A_i, so I would rather not reuse that
    Let $A_{\cup S} = \bigcup_{i\in S} A_i$ be the set of candidates approved by at least one voter in $S$, and let $W_S=W\cap A_{\cup S}$; then $|W_S|<\ell\le |T|$.
    We can assume without loss of generality that $T \subseteq A_{\cup S}$. Indeed, $S$ is Droop weakly $(\ell,T)$-cohesive if and only if it is Droop weakly $(\ell,T\cap A_{\cup S})$-cohesive, so we can replace $T$ with $T\cap A_{\cup S}$. 
    Let $S' = \{i\in S \mid \pi(i) \notin A_i\}$ be the set of voters in $S$ that are assigned to a candidate they do not approve. Note that $|T\setminus W_S|\ge |T|-|W_S|>0$, 
    and let $T' \subseteq T\setminus W_S$ be an arbitrary subset of $T\setminus W_S$ of size $|T| - |W_S|$. 
    
    Since $|S|>|T|\cdot\frac{n}{k+1}$ and $|\pi^{-1}(c)|=\frac{n}{k+1}$ for all $c\in W_S$, we can lower-bound the number of voters in $S'$ as
    \[
    |S'| \geq |S| - |W_S|\cdot\frac{n}{k+1} > \left(|T| - |W_S|\right)\cdot\frac{n}{k+1} = |T'|\cdot\frac{n}{k+1}.
    \] 
    Furthermore, each voter in $S'$ approves at least $\ell$ candidates in $T$, so she approves
    at least
    $\ell - |W_S|$ candidates in $T'$. Therefore, the total number of approvals given by voters in $S'$ to candidates in $T'$ is at least $|S'| \cdot (\ell - |W_S|) \geq |S'|$. Thus, by the pigeonhole principle there exists a candidate $c'\in T'$ who is approved by at least $\frac{|S'|}{|T'|} > \frac{n}{k+1}$ voters in $S'$. Now we deal with the two rules separately.

    For the Droop Monroe rule, by applying Lemma~\ref{lem:monroe-unsatisfied-voters} to $N_{c'}\cap S'$ we conclude that $c'\in W$, a contradiction with $c' \in T' \subseteq T\setminus W$.

    For the Droop Greedy Monroe rule, we note that, since the algorithm assigns exactly $\frac{n}{k+1}$ voters to each of the $k$ candidates in $W$, there are exactly $\frac{n}{k+1}$ unassigned voters. As $|S'|> |T'|\cdot\frac{n}{k+1} \geq \frac{n}{k+1}$, it cannot be the case that all voters in $S'$ remain unassigned. Consider the first point in time when the algorithm assigns a voter from $S'$ (say, $i$) to a candidate (say, $c$). By definition of $S'$ we have $c\not\in A_i$. 
    This means that the algorithm was unable to find $\frac{n}{k+1}$ active voters in $N_c$.
    However, $|N_{c'}\cap S'|\ge \frac{n}{k+1}$ and all voters in $N_{c'}\cap S'$ remain active at this point, a contradiction with the algorithm choosing $c$, since the algorithm always selects a candidate with the largest number of active approvers.

    In both cases we reach a contradiction, so we conclude that the Droop Monroe rule and Droop Greedy Monroe rule satisfy Droop-FPJR if $k+1$ divides $n$.
    \qed
\end{proof}

We note that the condition that $k+1$ divides $n$ is not an artifact of the proof: both Droop Monroe and Droop Greedy Monroe can violate Droop FPJR if it is not satisfied, even if $k$ divides $n$ (Proposition~\ref{prop:droop-monroe-not-droop-pjr} in the appendix). Moreover, just as for other rules defined in terms of a quota, the Hare versions of Monroe and Greedy Monroe do not satisfy Droop-FPJR, even if both $k$ and $k+1$ divide $n$ (Proposition~\ref{prop:monroe-not-droop-pjr} in the appendix). Interestingly, Hare Monroe satisfies the weaker Droop-JR axiom if $k$ divides $n$ (Proposition~\ref{prop:monroe-droop-jr} in the appendix), but this result does not extend to Hare Greedy Monroe (\Cref{prop:greedy-monroe-not-droop-jr} in the appendix).

\subsection{Priceable Rules}

Next, we consider rules 
that satisfy the priceability axiom;
this includes, in particular, SeqPhragm\'{e}n~\cite{brill_phragmen_2024}, the Maximin Support Method~\cite{MMS}, MES~\cite{MES}, and EES~\cite{kraiczy_streamlining_2025}. Intuitively, this axiom is satisfied by committees that can be purchased by the voters if all voters are given equal budgets and can spend them on candidates they approve.

\begin{definition}[Priceability~\cite{MES}]\label{def:priceability}
    A {\em price system} is a pair $(p, (p_i)_{i\in N})$ where $p\in{\mathbb Q}^+$ is a \emph{price} and $p_i:C\to{\mathbb Q}^+\cup\{0\}$ is the {\em payment function} of voter $i\in N$. 
    For each $i\in N$ the payment function $p_i$ satisfies (1) $p_i(c) =0$ for $c\notin A_i$, and (2) $\sum_{c\in C}p_i(c) \leq 1$. 
    A
    price system $(p, (p_i)_{i\in N})$ {\em supports} a committee $W$ if 
    \begin{itemize}
        \item $\sum_{i\in N} p_i(c) = p$ for each $c\in W$;
        \item $p_i(c) = 0$ for each $i\in N$, $c\notin W$; 
        \item For each candidate $c\notin W$, the remaining budget of the supporters of $c$ is at most $p$:\\ $\sum_{i\in N_c} (1 - \sum_{c\in W}p_i(c)) \leq p$ for each $c\notin W$.
    \end{itemize}

    A committee $W$ is {\em priceable} if it is supported by a price system. A voting rule $\mathcal R$ is {\em priceable} if it outputs priceable committees.
\end{definition}

%we will discuss are SeqPhragm\'{e}n, the Method of Equal Shares (MES), and Exact Equal Shares (EES). Notably, Brill et al.~\cite{brill_phragmen_2024} proved that SeqPhragm\'{e}n satisfies Droop-PJR. Subsequently, 

The notion of priceability was introduced by Peters and Skowron~\cite{MES}, 
who proved that any priceable size-$k$ committee provides PJR.
Subsequently, Kalayc{\i} et al.~\cite{FPJR} 
strengthened this result, showing that any priceable 
size-$k$ committee provides FPJR. We will now extend this result to Droop-FPJR.

We will use the following technical lemma.

\begin{lemma}\label{lem:xy}
 For any pair of positive integers $x, y\in\mathbb Z$ it holds that $xy+1\ge x+y$.   
\end{lemma}
\begin{proof}
    If $x=1$ or $y=1$, our claim is immediate. Now, suppose that $x, y\ge 2$. Then $xy+1>xy= \frac{xy}{2}+\frac{xy}{2}\ge \frac{2y}{2}+\frac{2x}{2}=x+y$.
    \qed
\end{proof}

\begin{theorem}
    \label{thm:priceable-droop-fpjr}
    Every priceable committee $W$ for an election with $k=|W|$ provides Droop-FPJR.
\end{theorem}

\begin{proof}
    Let $(p,(p_i)_{i\in N})$ be a price system for the committee $W$. For each $i\in N$ let $b_i = 1 - \sum_{c\in W}p_i(c)$ be $i$'s remaining budget. Assume for contradiction that $W$ does not provide Droop-FPJR. Then there is a Droop weakly $(\ell,T)$-cohesive group $S$ such that $|W \cap \bigcup_{i\in S}A_i| < \ell$. Just as in the proof of Theorem~\ref{thm:droop-monroe-droop-fpjr}, we can assume without loss of generality that $T \subseteq \bigcup_{i\in S}A_i$.  
    %Now let $W_T = W\cap T$ and $R = T\setminus W$, so that $W_T$ and $R$ partition $T$ into those candidates that are in $W$ and those that are not. 
    Let $W_S = W\cap \bigcup_{i\in S}A_i$ and $O = W\cap T$.
    
    Let $B$ be the sum, over all candidates  $c\in T\setminus O$,  of the remaining budgets of the supporters  of $c$. We start by lower-bounding $B$. To this end, we observe that 
    \begin{align}\label{eq:forB}
    \sum_{i\in S}\sum_{c\in W}p_i(c) = \sum_{c\in W}\sum_{i\in S}p_i(c) = \sum_{c\in W_S}\sum_{i\in S}p_i(c) \le  p\cdot |W_S|, 
    \end{align}
    where the second and third transition follow from the properties of the price system, namely,
    that $p_i(c)=0$ for all $i\in S$, $c\in W\setminus W_S$ and that $\sum_{i\in S}p_i(c)\le p$ for each $c\in W$. 
    Moreover, each member of $S$ approves at least $\ell$ candidates in $T$, and hence at least $\ell - |O| > 0$ candidates in $T\setminus O$. Hence, 
    \begin{align*}
        B = \sum_{c\in T\setminus O}\sum_{i\in N_c} b_i \geq (\ell - |O|)\sum_{i\in S} b_i
        = (\ell - |O|)\sum_{i\in S}(1- \sum_{c\in W}p_i(c))
        \geq (\ell - |O|)(|S| - p\cdot |W_S|), 
    \end{align*}
    where the last transition follows from Eq.~\eqref{eq:forB}.
    To upper-bound $B$, we use the fact that $\sum_{i\in N_c}b_i \le p$ for each $c\notin W$ and hence
    \begin{align*}
        B = \sum_{c\in T\setminus O}\sum_{i\in N_c} b_i 
        \leq p\cdot |T\setminus O|
        = p\cdot (|T| - |O|).
    \end{align*}
    Putting these two bounds together and dividing by $p$, we obtain
    \begin{equation}
        \label{eq:priceability_droopfpjr}
        |T| - |O| \geq (\ell - |O|)\left(\frac{|S|}{p} - |W_S|\right).
    \end{equation}
    We now consider two cases: (1) $p \leq \frac{n}{k+1}$ and (2) $p>\frac{n}{k+1}$. 
    
    Suppose first $p \leq \frac{n}{k+1}$.
    Then, as $|S| >  |T|\cdot\frac{n}{k+1}$, we obtain
     $\frac{|S|}{p} > |T|$ and hence
     $$
     (\ell - |O|)\cdot\left(\frac{|S|}{p} - |W_S|\right) > (\ell-|O|)\cdot (|T|-|W_S|)\ge \ell-|O|+|T|-|W_S|-1\ge |T|-|O|, 
     $$
     where the second transition uses \Cref{lem:xy} with $x=\ell-|O|>0$, $y=|T|-|W_S|>0$, and the last transition uses the observation that $|W_S|<\ell$.
     This is a contradiction with Eq.~\eqref{eq:priceability_droopfpjr}.
    %\textcolor{red}{I do not understand the logic in their proof as to why you get a contradiction from this.}

    On the other hand, suppose $p>\frac{n}{k+1}$.
    We will then show that $|W\setminus W_S|<k+1-\ell$.
    Indeed, we have
    $$
    |N\setminus S| = n - |S| < n - |T|\cdot\frac{n}{k+1} = \frac{(k+1-|T|)\cdot n}{k+1}.
    $$
    Moreover, under the price system $(p, (p_i)_{i\in N})$ only the voters in $N\setminus S$ can pay for candidates in $W\setminus W_S$, and the price of each candidate in $W\setminus W_S$ is $p$, so 
    $$
    |W\setminus W_S|\le \frac{|N\setminus S|}{p}< \frac{(k+1-|T|)\cdot n}{k+1}\cdot \frac{1}{p} < \frac{(k+1-|T|)\cdot n}{k+1} \cdot \frac{k+1}{n} = k+1-|T| \leq k+1-\ell.
    $$
    As $|W_S|\le \ell-1$, we obtain 
    $|W|=|W\setminus W_S|+|W_S|<k$, a contradiction.
    
     In both cases we reached a contradiction, so we conclude that every priceable committee of size $k$ provides Droop-FPJR.
    \qed
\end{proof}

SeqPhragm{\'e}n~\cite{brill_phragmen_2024} and the Maximin Support Method (MMS)~\cite{MMS} are iterative voting rules that always output size-$k$ committees, and Peters et al.~\cite{MES} show that their outputs are priceable (we omit the definition of SeqPhragm{\'e}n and MMS, as they are not relevant to the discussion). 
Thus, we obtain the following corollary.

\begin{corollary}
    \label{cor:seq-phragmen-droop-fpjr}
    SeqPhragm{\'e}n and MMS satisfy Droop-FPJR.
\end{corollary}

We note that Brill et al.~\cite{brill_phragmen_2024} directly show that SeqPhragm{\'e}n satisfies Droop-PJR.

Kalayc{\i} et al.~\cite{FPJR} show that EJR and FPJR are incomparable.
By combining Corollary~\ref{cor:seq-phragmen-droop-fpjr} with the facts that SeqPhragm\'{e}n fails EJR~\cite{brill_phragmen_2024}, while PAV satisfies Droop-EJR~\cite{janson_2018}, but not FPJR~\cite{FPJR}, we obtain a Droop quota equivalent of this result.

\begin{corollary}
\label{cor:ejr-fpjr-incomparable}
    Droop-EJR and Droop-FPJR are incomparable.
\end{corollary}

%\begin{proof}
%    PAV satisfies Droop-EJR~\cite{janson_2018}, but does not satisfy FPJR~\cite{FPJR}, meaning it also does not satisfy Droop-FPJR. On the other hard, SeqPhragm\'{e}n satisfies Droop-FPJR (Corollary~\ref{cor:seq-phragmen-droop-fpjr}), but does not satisfy EJR~\cite{brill_phragmen_2024}, meaning it also does not satisfy Droop-EJR.
%    \qed
%\end{proof}

Peters et al.~\cite{MES} show that MES always outputs priceable committees.
Kalayc{\i} et al.~\cite{FPJR} use this result,  together with the claim that any size-$k$ priceable committee satisfies FPJR, to conclude that MES satisfies FPJR. However, this proof strategy is problematic, since MES may output committees with fewer than $k$ members. Fortunately, this issue can be circumvented by combining MES with SeqPhragm{\'e}n: Peters et al.~\cite{MES} observe that if MES returns a committee $W$ with $|W|<k$, we can run SeqPhragm{\'e}n with starting budgets equal to the remaining budgets of all voters at the end of MES to select the remaining $k-|W|$ candidates, and the resulting committee will be priceable. Thus, we obtain the following corollary.

\begin{corollary}
    \label{cor:mes-seq-phragmen-droop-fpjr}
    If MES/EES are completed with SeqPhragm{\'e}n, then their outcomes provide Droop-FPJR.
\end{corollary}

In fact, we can show that the completion-by-Phragm{\'e}n trick is not necessary: we will now give a direct proof that MES/EES satisfy FPJR, whereas Droop MES/EES satisfy Droop FPJR.

Our proof relies on two lemmas. 
The first lemma is technical; for readability, we relegate its proof to Appendix~\ref{app:omit}.

\begin{lemma}
\label{lem:maximizer-of-inverse-variables}
    Consider a set of positive integers 
    $x_1 \geq x_2 \geq \hdots \geq x_{t}$ 
    such that for some $s, \tau\in \mathbb N$
    it holds that $x_i\in [s]$ for all $i\in [t]$ and 
    $\sum_{i\in [t]} x_i \geq s\tau$. Then
    \[
    \sum_{i\in [\tau]} \frac{1}{x_i} \leq \frac{t}{s}.
    \]
\end{lemma}

The second lemma can be thought of as a monotonicity property of MES/EES: it shows that if under these rules a subset of voters $S$ can afford to pay for a subset of candidates $T$, then in a bigger election with additional voters and approvals MES/EES will guarantee these voters a collective utility of $|T|$. We formulate this lemma so that it can be used with both the Hare quota and the Droop quota. 

\begin{lemma}
    \label{lem:mes-pjr-monotonicity}
     Consider an execution of MES/EES where the initial budget of each voter is $\beta$ and the cost of each candidate is~$\gamma$. Consider a group of voters $S\subseteq N$ and a set of candidates $D\subseteq \cup_{i\in S}A_i$, 
     and for each $c\in D$ let $x_c=|N_c\cap S|$ be the number of voters in $S$ who approve $c$.
     Then if $\sum_{c\in D}\frac{\gamma}{x_c} \leq \beta$,  
     MES/EES selects at least $|D|$ candidates from $\cup_{i\in S}A_i$, and each voter in $S$ spends at most $\sum_{c\in D}\frac{\gamma}{x_c}$ on the first $|D|$ candidates selected by MES/EES from $\cup_{i\in S}A_i$. 
\end{lemma}

\begin{proof}
    The proof proceeds by induction on the size of $D$. For the base case, suppose $D$ is a singleton, i.e., $D=\{c\}$. Then the affordability threshold of $c$ at the first iteration of MES/EES is at most $\frac{\gamma}{x_c} \leq \beta$, 
    since the voters in $S$ can purchase $c$ by splitting its cost evenly amongst them. 
    Thus, once the rule terminates, at least one candidate from $\cup_{i\in S}A_i$ is purchased: otherwise, each voter in $S$ still has $\beta$ dollars, so voters in $S$ can collectively afford $c$.
    Further, if the first candidate from $\cup_{i\in S}A_i$ that is purchased by MES/EES is $c$, each voter in $S$ spends at most $\frac{\gamma}{x_c}$ on this purchase. On the other hand, if MES/EES 
    purchases some candidate $d\in (\cup_{i\in S}A_i)\setminus\{c\}$ before $c$, then the affordability threshold of $d$
    is at most $\frac{\gamma}{x_c}$, so no voter in $S$ spends more than $\frac{\gamma}{x_c}$ on this candidate. Thus, our claim holds for $|D|=1$.

    For the inductive step, assume that the claim holds for $|D| = \delta$; we will prove it for $|D| = \delta+1$. Pick $c^*\in\arg\min\{x_c: c\in D\}$ and set $D'=D\setminus\{c^*\}$.
    %Let $D'$ be the set of $\delta$ candidates in $D$ with the highest $x_i$, with ties broken arbitrarily. 
    Then, by applying the inductive hypothesis to $D'$, we conclude that MES/EES purchases at least $\delta$ candidates from $\cup_{i\in S}A_i$. Let $W_\delta$ be the set of the first $\delta$ candidates from $\cup_{i\in S}A_i$ purchased by MES/EES; by the inductive hypothesis, each voter in $S$ spends at most $\sum_{c\in D'}\frac{\gamma}{x_c}=\sum_{c\in D}\frac{\gamma}{x_c}-\frac{\gamma}{x_{c^*}}\le \beta-\frac{\gamma}{x_{c^*}}$ on candidates in $W_\delta$. Therefore, after purchasing the candidates in $W_\delta$, each voter in $S$ has at least $\frac{\gamma}{x_{c^*}}$ dollars remaining. As $|D|>|W_\delta|=\delta$, the set $D\setminus W_\delta$ is non-empty; consider some candidate $d\in D\setminus W_\delta$. By the choice of $c^*$, we have $x_d\ge x_{c^*}$, so at this point the voters in $S$ can purchase $d$ at an affordability threshold of at most $\frac{\gamma}{x_{c^*}}$. The only reason why they may fail to do that is that some other candidate from $\cup_{i\in S}A_i$ is purchased; since MES/EES always pick a candidate with the lowest affordability threshold, in that case, too, the cost to each voter in $S$ is at most $\frac{\gamma}{x_{c^*}}$. Either way, MES/EES  purchases at least $\delta+1$ candidates, with no voter spending more than $\sum_{c\in D}\frac{\gamma}{x_c}$ on them, so our proof is complete.
    \qed
\end{proof}

We are now ready to prove that MES/EES satisfy FPJR even when they select fewer than $k$ candidates.

\begin{theorem}
    \label{thm:mes-fpjr}
    MES/EES satisfy FPJR, and Droop MES/EES satisfy Droop-FPJR. 
\end{theorem}
\begin{proof}
    Consider any (Droop) weakly $(\ell, T)$-cohesive group $S$. Let $x_i = |N_{c_i} \cap S|$ denote the number of voters in $S$ that approve a candidate $c_i\in T$. Relabel the candidates so that $x_1 \geq x_2 \geq\ldots\geq x_{|T|}$. We assume without loss of generality that $x_i \geq 1$ for each $c_i\in T$: if not, we can replace $T$ with $T\cap(\cup_{i\in S}A_i)$. By definition we also have $x_i \leq |S|$. Since $S$ is (Droop) weakly $(\ell, T)$-cohesive, each voter $i\in S$ approves at least $\ell$ candidates in $T$, so $x_1 + \cdots + x_{|T|} \geq \ell\cdot|S|$. Therefore, by applying Lemma~\ref{lem:maximizer-of-inverse-variables} with $s=|S|, t=|T|, \tau = \ell$, we obtain $\sum_{i\in [\ell]} \frac{1}{x_i} \leq \frac{|T|}{|S|}$. To conclude the proof, we will invoke \Cref{lem:mes-pjr-monotonicity} with $D=\{c_1, \dots, c_\ell\}$ and $\gamma=\frac{n}{k}$; to do so, we need to show that 
    $\sum_{i\in [\ell]}\frac{n}{k}\cdot\frac{1}{x_i}\le \beta$, where $\beta$ is the Hare (resp., Droop) MES/EES per voter budget.

    %The per voter cost to purchase a candidate $c_i\in T$ is $\frac{n}{kx_i}$.
    %Thus, if the first $\ell$ candidates in $T$ can
    %hence, even if a voter was involved in purchasing all of the first $\ell$ candidates, they would spend at most $\sum_{i\in [\ell]}\frac{n}{kx_i}\leq \frac{n}{k}\cdot \frac{|T|}{|S|}$.
    
    For the standard (Hare) MES/EES with per-voter budget $1$, since $S\ge\frac{n}{k}\cdot|T|$, we have
    $$
    \sum_{i\in[\ell]} \frac{n}{k}\cdot\frac{1}{x_i}\le
    \frac{n}{k}\cdot\frac{|T|}{|S|}\le 1.
    $$
    For Droop MES/EES with per-voter budget $\frac{(k+1)n}{kn+1}$, 
    we have $|S|>\frac{n}{k+1}\cdot|T|$;
    by \Cref{lem:size-of-cohesive-group}, this implies
    $|S|\ge\frac{n\cdot|T|+1}{k+1}$.
    Moreover, we have $|T|\le k$: indeed, $|T|\ge k+1$ implies $|S|>n$, which is impossible.
    Hence, we obtain
    \begin{align*}
    \sum_{i\in[\ell]} \frac{n}{k}\cdot\frac{1}{x_i}\le
        \frac{n}{k}\cdot \frac{|T|}{|S|} \leq \frac{n}{k}\cdot\frac{|T|\cdot(k+1)}{n\cdot|T|+1}
        = \frac{n\cdot|T|}{n\cdot|T|+1}\cdot\frac{k+1}{k}
        %= \frac{k+1}{k}\left(1-\frac{1}{n\cdot|T|+1}\right)\\
        %&\le
        %\frac{k+1}{k}\left(1-\frac{1}{n\cdot k+1}\right)=
        \le\frac{kn}{nk+1}\cdot\frac{k+1}{k}=\frac{(k+1)n}{kn+1}, 
    \end{align*} 
    where we use the fact that the function $\frac{nz}{nz+1}$ is monotonically increasing in $z$ for $z\ge 0$.
    
    In either case, we can invoke Lemma~\ref{lem:mes-pjr-monotonicity} to conclude that every outcome $W$ of (Droop) MES/EES satisfies $|W\cap \bigcup_{i\in S} A_i| \geq \ell$, which is what we wanted to prove.
    \qed
\end{proof}
%%%%%%%%%%%%%%%%%%%%%%%%%%%%%%%%%%%%%%%%%%%%%%%%%%%%%%%%%%%%%%%%%%%%%
\section{Full Justified Representation}\label{sec:fjr}
Full Justified Representation is known to be a challenging axiom to satisfy, even for the Hare quota. Indeed, there is only one voting rule known to satisfy Hare-FJR, namely, the Greedy Cohesive Rule (GCR), which is not known to be polynomial-time computable. Below, we give a formal definition of this rule; as it proceeds by identifying cohesive groups, we define two variants of this rule: Hare-GCR (which is identical to the GCR rule defined in prior work) and Droop-GCR.

\begin{definition}[Hare/Droop-Greedy Cohesive Rule (GCR)~\cite{FJR}] 
The rule starts with $W = \varnothing$ and all voters $v\in N$ marked as active. It proceeds iteratively.
At each step, it constructs a set $\mathcal T$ that consists of all triples $(\ell, T, S)$ with
$\ell\in [k]$, $T\subseteq C\setminus W$ and $S \subseteq N$ such that all voters in $S$ are active and $S$ is Hare (resp., Droop) weakly $(\ell, T)$-cohesive. If ${\mathcal T}=\varnothing$, it adds $\max\{0, k-|W|\}$ arbitrary candidates from $C\setminus W$ to $W$, and outputs $W$. Otherwise,
among all triples $(\ell, T, S)\in\mathcal T$ it identifies the ones with the largest value of $\ell$, 
and picks one with the smallest $|T|$ among these. It then adds all candidates in $T$ to $W$, and marks all voters in $S$ as inactive.   
\end{definition}

Peters et al.~\cite{FJR} show that Hare-GCR satisfies Hare-FJR. We will now adapt their proof to show that Droop GCR satisfies Droop-FJR. 

\begin{theorem}
    \label{thm:droop-gcr-droop-fjr}
    Droop GCR selects a committee of size~$k$ and satisfies Droop-FJR.
\end{theorem}
\begin{proof}
    First, we prove that Droop-GCR satisfies Droop-FJR. Let $W$ be an output of Droop-GCR, and assume for contradiction that there is a Droop weakly $(\ell,T)$-cohesive group $S$ such that $|A_i \cap W| < \ell$ for all $i\in S$. Since the algorithm terminated without adding $T$ to $W$, some member of $S$ must have been marked as inactive by the rule; let $i$ be the first such voter, and suppose
    that $i$ was marked as inactive as part of some Droop weakly $(\ell', T')$-cohesive group $S'$. When $(\ell', T', S')$ was chosen, all members of $S$ were still active, so $S'$ being chosen means that $\ell' \geq \ell$. But we also have $\ell > |A_i \cap W| \geq |A_i \cap T'| \geq \ell'$. Thus, we obtain $\ell > \ell'$, a contradiction. Hence, $W$ provides Droop-FJR.

    Next, we show that Droop-GCR outputs $k$ candidates. Suppose that Droop-GCR constructs $W$ by adding sets of candidates $T_1, \dots, T_r$, so that, after $T_r$ is added, in the next iteration the set $\mathcal T$ is empty. When $T_j$ is added, the algorithm marks more than $\frac{|T_j|n}{k+1}$ voters as inactive. As each voter is marked as inactive at most once, we have $\sum_{j=1}^r|T_j|<k+1$. Thus, after $T_r$ is added, we have $|W|\le k$, and in the next iteration the algorithm outputs a set of size exactly $k$. 
    \qed
\end{proof}

We note that considering Droop weakly $(\ell, T)$-cohesive groups in the definition of the rule is necessary for Droop FJR:
\Cref{thm:gjcr-gcr-not-droop-jr} shows that Hare GCR fails even the weaker Droop JR axiom (see appendix).

%%%%%%%%%%%%%%%%%%%%%%%%%%%%%%%%%%%%%%%%%%%%%%%%%%%%%%%%%%%%%%

\section{Experiments}\label{sec:experiments}
The Droop quota versions of the proportionality axioms we have introduced are by definition stronger than their Hare quota cousins. A natural question, then, is whether the Droop versions are actually harder to satisfy in practice. That is, for typical instances, how common is it for a committee to provide, say, Hare-EJR, but not Droop-EJR? We study this question from three different perspectives, using a variety of sampling models.
For reasons of computational efficiency, we focus on easy-to-verify axioms, namely, Hare/Droop-JR and Hare/Droop-EJR+ (for a discussion of complexity of verification, see~\cite{EJR+,aziz_2017,aziz_2018,FPJR}).
Our experiments indicate that the Droop versions of these axioms are satisfied meaningfully less frequently than the Hare versions.

%To address this question, we use the \emph{abcvoting} python library~\cite{joss-abcvoting} 

\medskip

\noindent{\bf Experiment 1\ }
In our first experiment, 
we ask what is the probability that a randomly selected committee satisfies an axiom; this experimental design was used in several prior works~\cite{EJR+,szufa_how_2022}. 

We use three sampling models: Resampling, Noise and Truncated Urn~\cite{szufa_how_2022} (see Appendix~\ref{app:exp} for their descriptions). All three sampling models are parameterized by a value $p$. For each model, we test four values of $p$: 0.2, 0.4, 0.6, and 0.8. The Resampling and Noise models are additionally parameterized by a value $\phi$, and the Truncated Urn model is parameterized by a value $\alpha$; for both $\phi$ and $\alpha$
 we test 100 values: each value from 0.01 to 1, in increments of 0.01.  We use $500$ voters, $50$ candidates and sample a random committee of size $10$. For each  combination of parameters, we run 400 repetitions. Our parameter settings replicate those 
 of Brill and Peters~\cite{EJR+} except for the number of voters: their experiments have $n=100$, while we set $n=500$. We do this because Hare and Droop quotas are the same for $n=100, k=10$: $\frac{100}{10}=\lfloor\frac{100}{11}\rfloor+1$; in contrast, for $n=500$ we get $\frac{500}{10}=50$, $\lfloor\frac{500}{11}\rfloor+1=46$.
 For each parameter combination we plot the fraction of committees that satisfy JR, Droop-JR, EJR+, and Droop-EJR+. 

%\vspace{-.5cm}
\begin{figure}
    \centering
    \includegraphics[width=\textwidth]{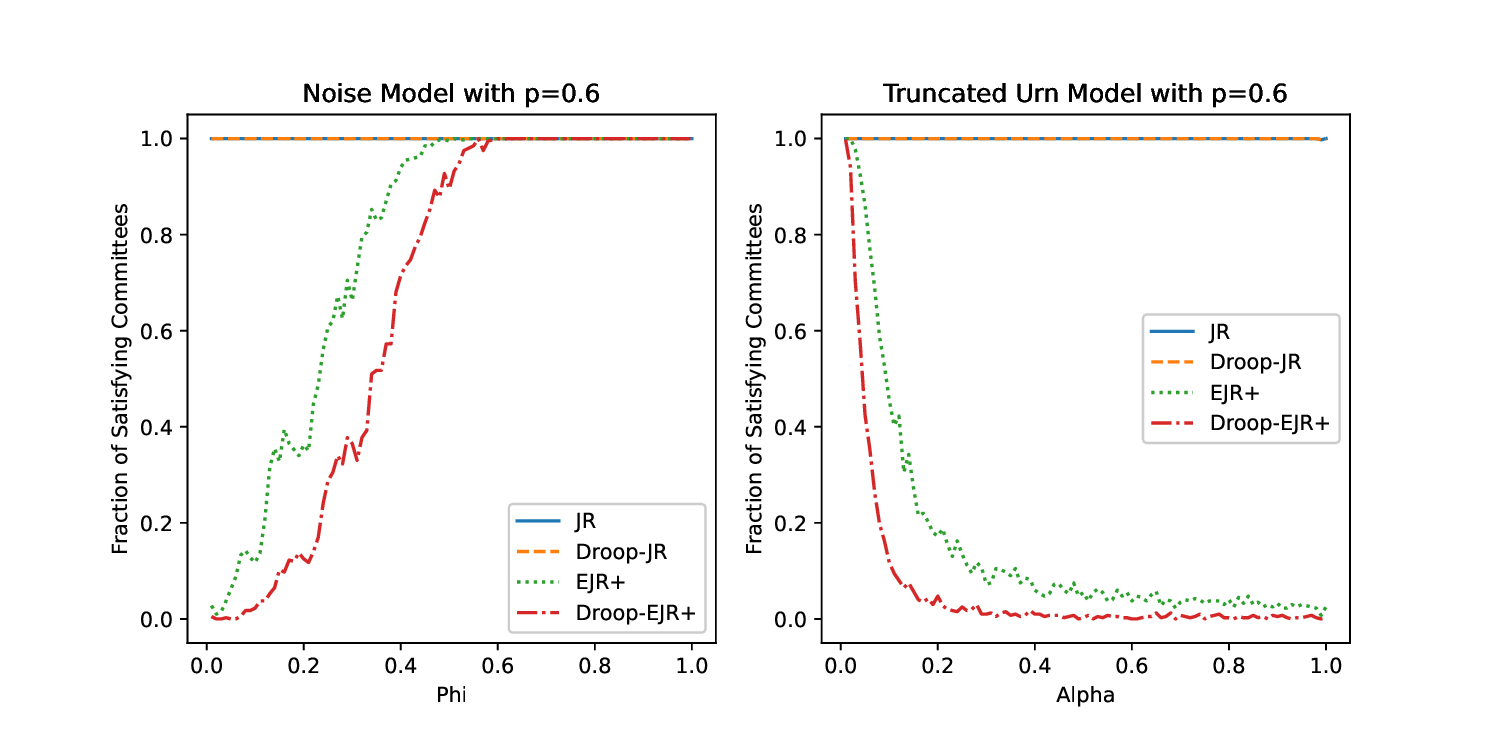}
    \caption{Experiment 1. The line for JR is not visible, as it coincides with the line for Droop-JR.}
    \label{fig:exp1}
\end{figure}

\medskip

\noindent{\bf Experiment 2\ }
Our second experiment uses the same sampling models and parameter settings. However, instead of selecting a committee at random, we generate a committee using the standard version of MES. Since MES always satisfies JR and EJR+, we only plot the fraction of committees that satisfy Droop-JR and Droop-EJR+.

\medskip

\noindent{\bf Experiment 3\ }
Our final experiment tests how the gap between satisfiability of the Hare and Droop axioms changes with the number of candidates and the size of the committee, $k$. We simulate elections following the  $p$-Impartial Culture model, where each voter approves each candidate independently with probability $p$ \cite{bredereck_experimental_2019}. This model can be viewed as a special case of the resampling model with $\phi = 1$. We use this model instead of the ones from the previous experiments, because it is parameterized by a single value, allowing for easier visualization. We test 100 values for $p$: each value from 0.01 to 1, in increments of 0.01. We set $n=100$, and test three values for the number of candidates: 50, 100, and 200. For the committee size $k$ we test values 1 to 9; we do not consider larger values of $k$, because for $k\geq \sqrt{n} = 10$ the Hare and Droop quotas are either identical or very close. For each of the combinations of parameters, we run 500 repetitions. We plot the fraction of committees that satisfy JR, Droop-JR, EJR+, and Droop-EJR+.

\smallskip

\noindent{\bf Results\ }
The full results of the experiments are given in Appendix~\ref{app:exp}:
 \Cref{fig:resampling-random-committee,fig:noise-random-committee,fig:truncated-urn-random-committee} for Experiment~1, \Cref{fig:resampling-mes-committee,fig:noise-mes-committee,fig:truncated-urn-mes-committee} for Experiment~2, and 
\Cref{fig:IC-50,fig:IC-100,fig:IC-200} for Experiment 3.
 
In Experiment~1, for many of the parameter settings, JR and Droop-JR are easily satisfied, with Droop-JR being slightly more demanding, especially for $p=0.2$ (see \Cref{fig:exp1} for a representative sample of results). On the other hand, EJR+ is satisfied far less often, and Droop-EJR+ is satisfied by far the least, especially for $p=0.8$. This offers evidence that Droop-EJR+ is significantly more difficult to satisfy than EJR+ (which is already quite demanding~\cite{EJR+}).
Thus, Droop-EJR+ serves as a powerful test of proportionality of committees.

\Cref{fig:exp2} offers a glimpse of results for Experiment~2. Although it is possible for MES to output a committee that fails Droop-JR, this rarely occurs in our experiments. 
%Only in the truncated urn model with $p=0.2$ are there a few instances where the MES outcome does not provide Droop-JR. 
In contrast, there is a large range of parameter values where MES outcomes rarely or never provide Droop-EJR+. This is especially true in the $p=0.8$ setting of the Resampling and Truncated Urn models. These results suggest that the standard MES rule cannot be relied on to provide Droop-EJR+ in practice; rather, one needs to use MES with a larger budget, as specified in Theorem~\ref{thm:droop-mes-droop-ejrplus}.

The results of Experiment 3 (see \Cref{fig:IC-50,fig:IC-100,fig:IC-200} in Appendix~\ref{app:exp})  %Note that the $y$-axes do not range from 0 to 1 on these charts. 
indicate that, even for relatively large values of $k$ there is a non-trivial range of parameters for which there is a meaningful difference between EJR+ and Droop-EJR+, even though  
$\frac{1}{k}$ and $\frac{1}{k+1}$ get closer as $k$ increases.

%\vspace{-.8cm}
\begin{figure}
    \centering
    \includegraphics[width=\textwidth]{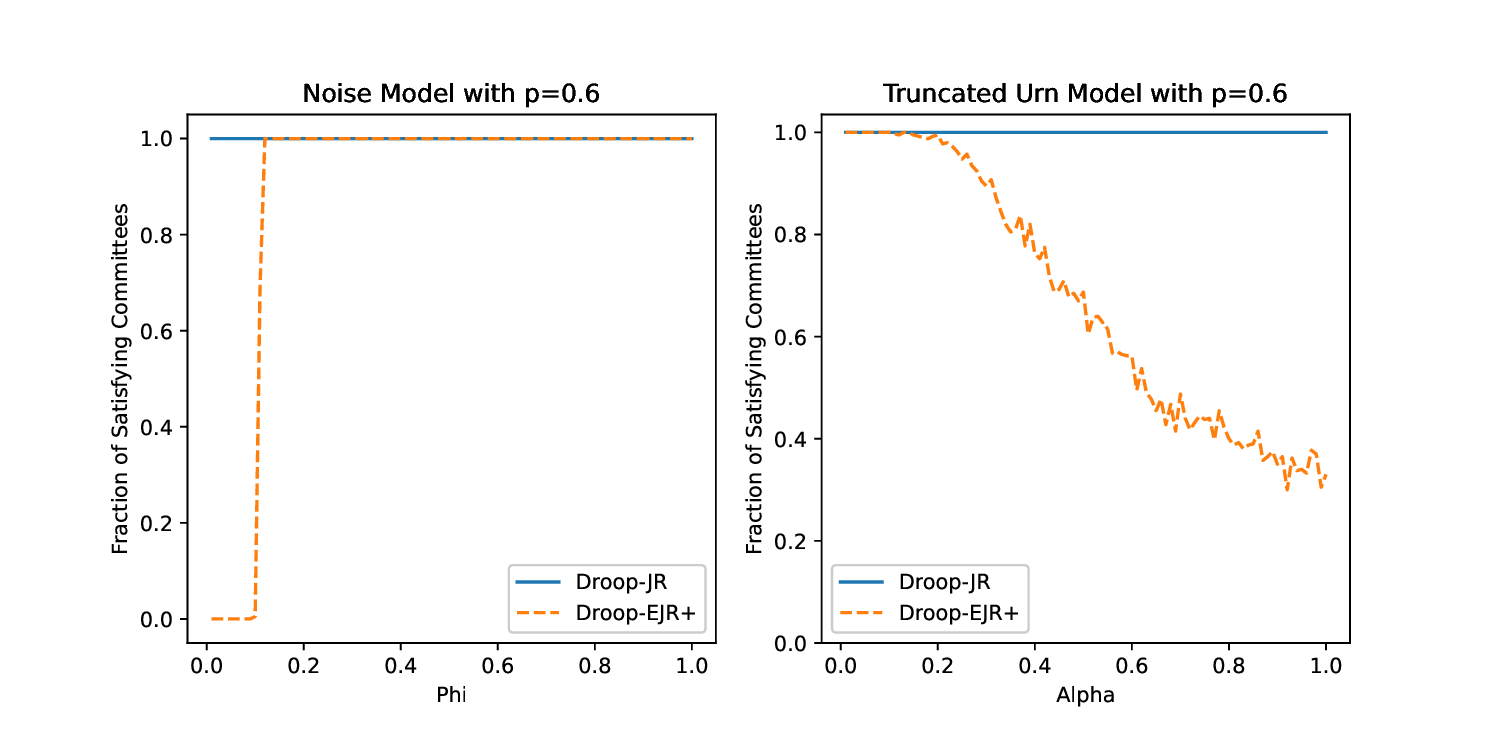}
    \caption{Experiment 2. Only the lines for Droop-JR and Droop-EJR+ are shown, as MES satisfies JR and EJR+.}
    \label{fig:exp2}
\end{figure}
 
%that for all parameter settings, all four axioms are satisfied very frequently, with satisfaction being less frequent as $k$ increases and as the number of candidates in the election increases. We expected that the Droop axioms should be most difficult to satisfy compared to the Hare versions when $k$ is small, because as $k$ grows, $\frac{1}{k}$ and $\frac{1}{k+1}$ get closer together. However, this is not the behavior we observe.

%\vspace{-0.8cm}

\section{Conclusions and Future Work}\label{sec:concl}
We have put forward Droop versions of the seven Justified Representation axioms that have been studied in the literature. For each axiom, we prove analogs of all major satisfiability results for voting rules from the Hare setting. While some of our proofs are simple adaptations of relevant proofs from prior work, others require more careful analysis or additional insights; this is the case, e.g., for Theorems~\ref{thm:pav-ejrplus} and~\ref{thm:droop-gjcr-droop-ejrplus}. Prior to our work, the strongest proportionality axioms known to be satisfiable were EJR+ and FJR. Our work advances this frontier by showing that the Droop versions of these axioms are always satisfiable. Furthermore, our experiments give evidence that in practice Droop-EJR+ is much harder to satisfy than EJR+.

Future work should try to fill in the remaining unknown entries in Table~\ref{tbl:results}. Furthermore, a natural direction for future work is to recover other known results from the Hare setting, such as for average satisfaction~\cite{aziz_2018}, and for hardness of verification~\cite{EJR+,aziz_2017,aziz_2018,FPJR}. It would also be natural to explore the use of the Droop quota in the burgeoning field of proportional clustering~\cite{li-clustering,micha-clustering,aziz-clustering,kellerhals-clustering}. Finally, there is a need for a more extensive set of experiments related to the Droop quota. In order to get a fuller picture, it would be interesting to repeat the experiments from this paper with FPJR and FJR. Furthermore, there should be experiments performed using data from real-world elections, which is  readily available from the PrefLib repository~\cite{MaWa13a}.

\bibliographystyle{splncs04}
\bibliography{refs.bib}

\appendix
\section{Omitted Proofs}\label{app:omit}

\noindent{\bf Lemma~\ref{lem:maximizer-of-inverse-variables}.\ }
{\em
    Consider a set of positive integers 
    $x_1 \geq x_2 \geq \hdots \geq x_{t}$ 
    such that for some $s, \tau\in \mathbb N$
    it holds that $x_i\in [s]$ for all $i\in [t]$ and 
    $\sum_{i\in [t]} x_i \geq s\tau$. Then
    \[
    \sum_{i\in [\tau]} \frac{1}{x_i} \leq \frac{t}{s}.
    \]
}    
\begin{proof}
    We first establish an auxiliary lemma.
    \begin{lemma}
    \label{lem:transferring-increases-sum}
    For all $a\geq b \geq 1$ we have
    \[
    \frac{1}{a+1} + \frac{1}{b-1} > \frac{1}{a} + \frac{1}{b}.
    \]
    \end{lemma}

\begin{proof}
We have $(a+1)(b-1)=ab+b-a-1<ab$ and hence
\[
    \frac{1}{a+1} + \frac{1}{b-1} =\frac{a+b}{(a+1)(b-1)}> \frac{a+b}{ab} = \frac{1}{a} + \frac{1}{b}.
    \]
    \qed
\end{proof}
We are now ready to present the proof of \Cref{lem:maximizer-of-inverse-variables}.
Consider any positive integers $x_1 \geq x_2 \geq \cdots \geq x_{t}$ that satisfy the condition in the statement of the lemma. 

Suppose first that $\tau=1$. Since $x_1\ge x_i$ for all $i\in [t]$, we have $x_1\ge\frac{1}{t}\cdot\sum_{i\in [t]}x_i\ge \frac{s}{t}$, and hence $\frac{1}{x_1} \leq \frac{t}{s}$, and our claim follows. 
    
Therefore, for the remainder of the proof we assume $\tau\geq 2$. We will now modify our variables in a way that does not decrease the objective $\sum_{i\in [\tau]} \frac{1}{x_i}$, and then show that our bound applies to the modified variables. Intuitively, we will transfer as much mass as possible from the later variables in $x_1, \dots, x_\tau$ to the earlier ones, while respecting the ordering and the constraints $x_i \leq s$ for all $i\in [t]$, and use Lemma~\ref{lem:transferring-increases-sum} to show that this can only increase our objective. 

Formally, our transformation proceeds as follows. 
We define $x_0=s$, $x_{t+1}=1$. Then, at each iteration we
check if there are variables $x_j, x_\ell$ with $1\le j<\ell\le \tau$ such that $x_{j-1}>x_j$ and $x_\ell>x_{\ell+1}$; if yes, we set $x_j:=x_j+1$, $x_\ell:=x_\ell-1$.
We repeat this step until no such pair can be found. Clearly, each step
of our transformation does not change the sum $\sum_{i\in [\tau]}x_i$, and the constraint
$s\ge x_1\ge\dots\ge x_t\ge 1$ remains satisfied. Moreover, since $x_j\ge x_\ell$ at the start of the step, by \Cref{lem:transferring-increases-sum} our transformation can only increase the sum $\sum_{i\in [\tau]}\frac{1}{x_i}$.
Thus, it suffices to show that $\sum_{i\in [\tau]}\frac{1}{x_i}\le\frac{t}{s}$ after the transformation.

If after the transformation we have $x_\tau=s$, our constraints imply $x_i=s$ for all $i\in [\tau]$ and it follows that $\sum_{i\in [\tau]}\frac{1}{x_i}=\frac{\tau}{s}\le\frac{t}{s}$, so we are done. Thus, suppose that after the transformation we have $x_\tau<s$. Then we claim that there exists a $j$ with $0\le j<\tau$ such that $x_1=\dots=x_j=s$ and $x_{j+2}=\dots=x_\tau$, with $x_\tau\le x_{j+1}< s$. Indeed, let $j=\max\{\ell\in [\tau]: x_\ell=s\}$; by our assumption, $j<\tau$. If $j\ge\tau-2$, our claim is vacuously true. On the other hand, suppose that $j< \tau-2$. If it is not the case that $x_{j+2}=\dots=x_{\tau}$, then for $\ell'=\max\{\ell\in[\tau]: x_\ell>x_\tau\}$ the variables 
$x_{j+1}$, $x_{\ell'}$ satisfy the constraints $x_j>x_{j+1}$, $x_{\ell'}>x_{\ell'+1}$, a contradiction with the termination condition for our transformation. This establishes our claim.

Now, we lower-bound $x_{j+1}$ by observing that
$$
x_{j+1}=\sum_{i\in [t]}x_i-\sum_{i\in [j]}x_i-\sum_{i=j+2}^t x_i \ge s\tau-sj-(t-j-1)x_\tau.
$$
Consequently, 
\begin{align}\label{eq:xj+1}
x_\tau - x_{j+1}\le 
x_\tau - s\tau+sj+(t-j-1)x_\tau = 
(t-j)x_\tau-(\tau-j)s, 
\end{align}
and we obtain
\begin{align*}
\sum_{i=1}^\tau\frac{1}{x_i} &= \sum_{i\in [j]}\frac{1}{x_i}+\frac{1}{x_{j+1}}+\sum_{i=j+2}^\tau\frac{1}{x_i} 
=
\frac{j}{s} + \frac{1}{x_{j+1}}+\frac{\tau-j-1}{x_\tau} =
\frac{j}{s} +\frac{1}{x_{j+1}}-\frac{1}{x_\tau}+\frac{\tau-j}{x_\tau}\\ 
&=
\frac{j}{s} +\frac{x_\tau-x_{j+1}}{x_\tau\cdot x_{j+1}} + \frac{(\tau-j)s}{x_\tau\cdot s}
\le
\frac{j}{s} + \frac{x_\tau-x_{j+1}}{x_\tau\cdot s} + \frac{(\tau-j)s}{x_\tau\cdot s}\le
\frac{j}{s} + \frac{(t-j)x_\tau}{x_\tau\cdot s} =\frac{j}{s} + \frac{t-j}{s} = \frac{t}{s}, 
\end{align*}
where the first inequality uses the facts that $x_{j+1}\le s$ and $x_\tau-x_{j+1}\le 0$, and the second inequality follows from~\eqref{eq:xj+1}.
    \qed
\end{proof}

\section{Negative Results}\label{app:neg}
In this section, we provide a number of (mostly) negative results. Most of these results consider a rule that satisfies the Hare version of a proportionality axiom and show that this rule fails the Droop version of the same axiom, thereby justifying our modifications of these rules.

Our first result explains why we need to run MES/EES with inflated budget to guarantee Droop-EJR+: if we use the standard budget, these rules are not even guaranteed to satisfy the much weaker Droop-JR axiom.

\begin{proposition}
    \label{prop:mes-not-droop-jr}
    MES/EES do not satisfy Droop-JR.
\end{proposition}

\begin{proof}
    The basic idea of the proof is that the size requirement for a Droop 1-cohesive group $S$ does not guarantee that $S$ will be able to afford their candidates. In particular, recall that when MES/EES are used, each voter gets a budget of 1, and each candidate costs $\frac{n}{k}$. The cohesiveness requirement only enforces that the group members have more than $\frac{n}{k+1}$ dollars, but they need $\frac{n}{k}$ dollars to buy their jointly approved candidate.
    
    Concretely, consider an election with $n=7$ voters and target committee size $k=2$. There are three voters who approve candidate $a$ only, and four voters who approve candidate $b$ only. When using MES/EES on this instance, each candidate costs $\frac{7}{2} = 3.5$. The voters in $N_a$ collectively have only 3 dollars to spend, so they cannot afford $a$.

    However, the group $N_a$ is Droop 1-cohesive: they have one jointly approved candidate, $a$, and $|N_a|=3 > \frac{7}{2+1}$. Therefore, in order to satisfy Droop-JR, MES/EES would have to output an outcome in which at least one voter in $N_a$ gets one of their approved candidates elected. However, as we noted, voters in $N_a$ only approve $a$, and MES/EES will not elect $a$. Thus, MES/EES do not satisfy Droop-JR.
    \qed
\end{proof}

Our next result shows that if we use the standard (Hare) versions of Monroe and Greedy Monroe, 
we may fail to satisfy Droop-PJR (and hence Droop-FPJR). The reader may wonder if this is caused by divisibility  issues (i.e., whether $n$ is divisible by $k$ or by $k+1$), but our proof demonstrates that this is not the case.

\begin{proposition}
    \label{prop:monroe-not-droop-pjr}
    Monroe and Greedy Monroe do not satisfy Droop-PJR, even when $k$ divides $n$ or $k+1$ divides $n$.
\end{proposition}

\begin{proof}
    First, we give an example with $k|n$. Consider an instance with $n=9, k=3$ where the first 7 voters approve candidates $c_1,c_2,c_3$, while the last two voters approve candidate $c_4$. Let $S$ be the group that consists of the first 7 voters. Then $S$ is Droop 3-cohesive: the voter in $S$ jointly approve three candidates and $|S|=7 > \frac{3\cdot9}{3+1}$. Therefore, the only outcome that provides Droop-PJR for this instance is $\{c_1,c_2,c_3\}$. However, the Monroe score of this committee is only 7, while the Monroe score of $\{c_1,c_2,c_4\}$ is 8. Furthermore, the Greedy Monroe rule would first select two candidates in $\{c_1, c_2, c_3\}$; assume without loss of generality that it selects $c_1$ and $c_2$. It would assign three voters from $S$ to $c_1$, then three more voters from $S$ to $c_2$. Finally, $c_4$ would be the remaining candidate with the most approvals from active voters, so Greedy Monroe would select it. Thus, neither Monroe nor Greedy Monroe are guaranteed to satisfy Droop-PJR when $k$ divides $n$.

    Now, we give an example with $(k+1)|n$. Consider an instance with $n=15, k=2$ where the first 11 voters approve candidates $c_1,c_2$ and the remaining 4 voters approve $c_3$. Let $S$ be the group that consists of the first 11 voters. Then $S$ form a Droop $2$-cohesive group: the voters in $S$ jointly approve 2 candidates, and $|S|=11 > \frac{2\cdot 15}{2+1}$. Therefore, the only outcome that provides Droop-PJR for this instance is $\{c_1, c_2\}$. However, the Monroe score of this committee is only 11, while the Monroe score of $\{c_1,c_3\}$ is 12. Furthermore, Greedy Monroe would first select one of $c_1, c_2$ and assign 8 of the voters in $S$ to it. Then, of the remaining two candidates, $c_3$ receives more approvals from active voters, so Greedy Monroe would select it. Therefore neither Monroe nor Greedy Monroe are guaranteed to satisfy Droop-PJR when $k+1$ divides $n$.
    \qed
\end{proof}

\Cref{prop:monroe-droop-jr} is the only positive result in this section, and offers a somewhat surprising observation: while Monroe fails Droop-PJR, it nevertheless satisfies a weaker Droop axiom, namely, Droop-JR, as long as $k$ divides $n$. 
We provide no matching negative result when $k$ does not divide $n$, so what happens in that regime remains an open question.

\begin{proposition}
    \label{prop:monroe-droop-jr}
    Monroe satisfies Droop-JR if $k$ divides $n$.
\end{proposition}

\begin{proof}
    Assume for contradiction that on some instance $(C, N, \calA, k)$ the Monroe rule selects a Hare valid assignment $(W, \pi)$ such that the committee $W$ fails Droop-JR. Then there exists a group of voters $S\subseteq N$ with size $|S| > \frac{n}{k+1}$ and $|\bigcap_{i\in S}A_i| \neq\varnothing$ such that $A_i \cap W = \varnothing$ for all $i\in S$. Note that voters in $S$ do not approve the candidates assigned to them by $\pi$, i.e., for each $i\in S$ we have $\pi(i)\not\in A_i$. By the pigeonhole principle, there must be some $c\in W$ such that $|\pi^{-1}(c)\cap S|\ge\frac{|S|}{k}$. Let $N'=\pi^{-1}(c)\setminus S$; since $|\pi^{-1}(c)|=\frac{n}{k}$, we have 
    $$
    |N'|= \frac{n}{k}-\frac{|S|}{k}<\frac{n}{k}-\frac{n}{k(k+1)}=\frac{n}{k+1}.
    $$
    Pick a candidate $c'\in\bigcap_{i\in S}A_i$, let $W'=W\cup\{c'\}\setminus\{c\}$, and consider a Hare valid assignment $(W', \pi')$ constructed as follows. Let $S'\subseteq S$ be a subset of $S$ of size $\min\{\frac{n}{k}, |S|\}$ that contains all voters in $\pi^{-1}(c)\cap S$. The mapping $\pi'$ assigns all voters in $S'$ to $c$, coincides with $\pi$ on all other voters in $N\setminus\pi^{-1}(c)$, and assigns all voters 
    from $N'$ to candidates in $W'$
    so that the resulting assignment is valid. 
    The voters in $S'$ prefer the new assignment to the old one, the voters in $N'$ may prefer the old assignment to the new one, and all other voters are indifferent between the two assignments. As we have argued that $|N'|<\frac{n}{k+1}< |S'|$, it follows that the new assignment has a higher Monroe score, a contradiction.
    \qed
\end{proof}

Interestingly, \Cref{prop:monroe-droop-jr} does not extend to Greedy Monroe, as shown in the following proposition. 

\begin{proposition}
    \label{prop:droop-monroe-not-droop-pjr}
    Droop Monroe and Droop Greedy Monroe do not satisfy Droop-PJR if $k+1$ does not divide $n$, even if $k$ divides $n$.
\end{proposition}

\begin{proof}
    Consider an election with $n=21,k=7$, where $A_i = \{c_1,\hdots, c_6\}$ for voters $i=1,\dots,16$, $A_{17} = A_{18} = \{c_7\}, A_{19} = A_{20} = \{c_8\}$, and $A_{21} = \{c_9\}$. Note that voters $1,\dots,16$ form a Droop $6$-cohesive group, since they jointly approve 6 candidates and $16 > \frac{6\cdot21}{7+1}$. Therefore, for a committee to provide Droop-PJR, it must contain all of $\{c_1,\dots,c_6\}$. However, both Droop Monroe and Droop Greedy Monroe will only select five of these candidates. Indeed since $21\mod (7+1) = 5$, a Droop valid assignment has $5$ candidates assigned to $3$ voters each, and $2$ candidates assigned to $2$ voters each (with two voters assigned to the dummy candidate $d$). Droop Greedy Monroe will select $\{c_1,\hdots,c_5\}$ first, assigning three voters to each of them. Then $c_7$ and $c_8$ will be selected because they have two supporters, while $c_6$ and $c_9$ only has one. For Droop Monroe, we note that $W = \{c_1,\hdots,c_6, c_7\}$ has Monroe score of 18, while $W = \{c_1,\hdots,c_5, c_7, c_8\}$ has Monroe score of 19. Therefore neither Droop Monroe nor Droop Greedy Monroe satisfy Droop-PJR.
    \qed
\end{proof}

\begin{proposition}
\label{prop:greedy-monroe-not-droop-jr}
    Depending on the way that ties are broken when assigning voters to a candidate they do not approve, Greedy Monroe may not satisfy Droop-JR, even if both $k$ and $k+1$ divide $n$.
\end{proposition}

\begin{proof}
    Consider an election with $n=100$ and $ k=4$. Note that $k$ and $k+1$ both divide $n$. The voters
    are partitioned into five pairwise disjoint groups $N_1, \dots, N_5$, so that voters in the group $N_i$ approve candidate $c_i$ for $i=1, \dots, 5$, and $N_1\cup\dots\cup N_5=N$.
    The sizes of these groups are $|N_1|=25$, $|N_2|=22$, $|N_3|=19$, $|N_4|=13$, $|N_5|=21$.
    For a committee to provide Droop-JR, it must select $c_5$, since it is jointly approved by voters in $N_5$, and $|N_5|=21 > \frac{100}{4+1}$. However, Greedy Monroe may select $c_1, c_2, c_3, c_4$, in this order. Indeed, $c_1$ has the largest number of supporters, so it is selected, and voters in $N_1$ are assigned to it. Next, $c_2$ is chosen, and Greedy Monroe may choose to assign all voters in $N_2$ as well as three voters from $N_5$ to it. At this point, $c_3$ would be the candidate with the largest number of approvals from active voters, so Greedy Monroe would choose it, and it may assign all voters in $N_3$ as well as six voters from $N_5$ to it. If that happens, $c_4$ will be the final candidate to be selected. Under this scenario, $c_5$ is not chosen, so Greedy Monroe does not satisfy Droop-JR, even when both $k$ and $k+1$ divide $n$.
    \qed
\end{proof}

The final result of this section offers the strongest evidence that the Droop quota is much more demanding than the Hare quota: we consider two voting rules that satisfy very strong Hare proportionality axioms (namely, GJCR, which satisfies Hare-EJR+, and GCR, which satisfies Hare-FJR) and show that they fail a very weak Droop axiom, namely, Droop-JR. In particular, this shows that Hare-EJR+ and Hare-FJR do not imply Droop-JR.

\begin{proposition}
    \label{thm:gjcr-gcr-not-droop-jr}
    GJCR and GCR do not satisfy Droop-JR.
\end{proposition}
\begin{proof}
     Consider an election with $n=3,k=1$, where voters 1 and 2 approve $a$ and voter 3 approves $b$. Since $2 > \frac{3}{1+1}$, voters 1 and 2 form a Droop 1-cohesive group. Thus, the only outcome that provides Droop-JR is $\{a\}$. However, GJCR will not select any candidates in its main loop, since in never considers $\ell>1$ (indeed, for $\ell>1$ we have $\frac{\ell n}{k} > n$), and there is no group of size $\frac{n}{k} = 3$ that jointly approves a candidate. Similarly, GCR will not select any candidates in its main loop, since the maximum size that $T$ can be is 1 (otherwise $\frac{|T|\cdot n}{k} > n$), and there is no group of size $\frac{n}{k}$ that jointly approves a candidate.
    \qed
\end{proof}

%%%%%%%%%%%%%%%%%%%%%%%%%%%%%%%%%%%%%%%%%
\section{Experiments}\label{app:exp}
In this section, we provide additional information about our sampling models and full experimental results. For our experiments, we use the \emph{abcvoting} Python library~\cite{joss-abcvoting} 

\subsection{Sampling Models}

\paragraph{Resampling Model} The resampling model is parameterized by two values: $p$ and $\phi$. We start by randomly sampling a central ballot $B$ that contains $\lfloor p\cdot m\rfloor$ approvals, where $m$ is the number of candidates in the election. Then for each voter $i$ we set $A_i = B$. Finally, with probability $1-\phi$ we resample $i$'s ballot and with probability $\phi$ we leave it as is. 

\paragraph{Noise Model} The noise model is parameterized by two values: $p$ and $\phi$. We start by randomly sampling a central ballot $B$ as in the resampling model. Then each voter's ballot is generated by sampling from the space of ballots with probability proportional to $\phi^d$, where $d$ is the Hamming distance between the ballot and $B$.

\paragraph{Truncated Urn Model} The truncated urn model is a truncated version of the P{\'o}lya-Eggenberger Urn Model. The model is parameterized by two values: $p$ and $\alpha$. We start with an urn that contains all $m!$ linear orderings of the candidates. For each voter we generate their ballot by drawing an ordering $r$ from the urn and truncating it so as to approve the first $\lceil p\cdot m\rceil$ candidates. Then we return $\alpha m!$ copies of $r$ to the urn.

\subsection{Full Results}
In Experiment 3 (\Cref{fig:IC-50,fig:IC-100,fig:IC-200}), for readability purposes we do not plot JR and Droop-JR, because the JR line almost exactly overlaps the EJR+ line, and the Droop-JR line almost exactly overlaps the Droop-EJR+ line.
\begin{figure}
    \centering
    \includegraphics[width=\textwidth]{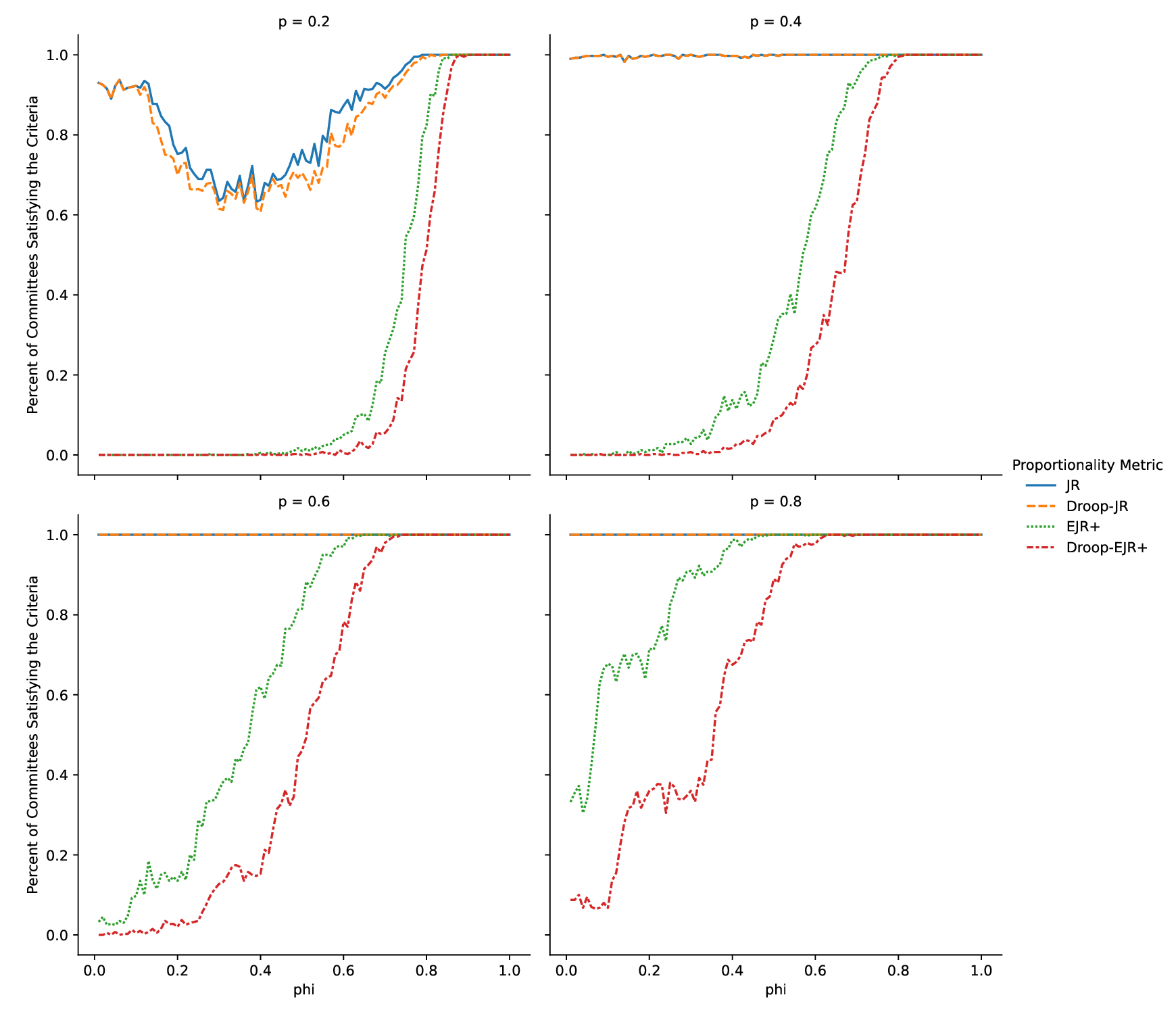}
    \caption{Results for the Resampling Model with a randomly generated committee.}
    \label{fig:resampling-random-committee}
\end{figure}

\begin{figure}
    \centering
    \includegraphics[width=\textwidth]{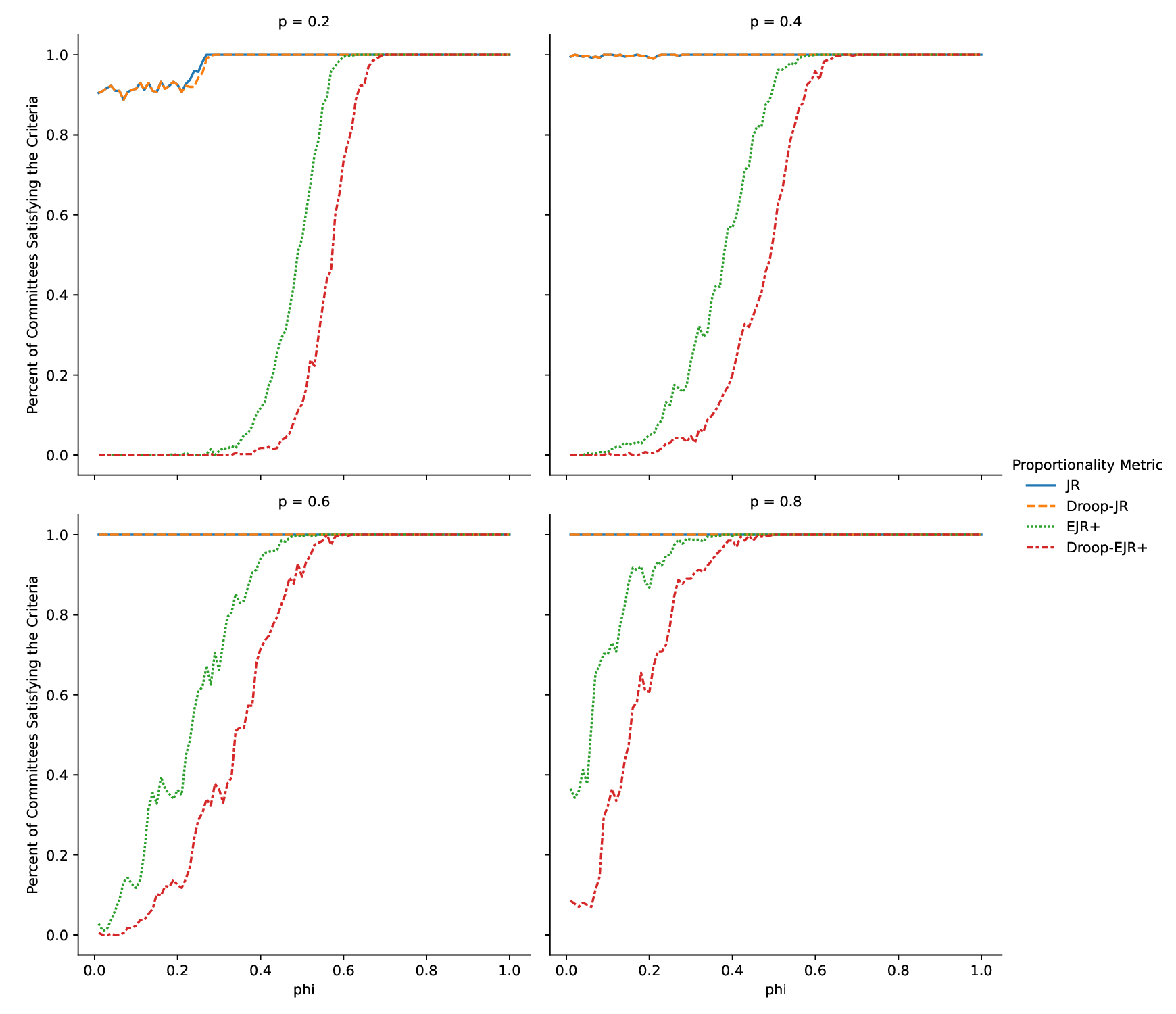}
    \caption{Results for the Noise Model with a randomly generated committee.}
    \label{fig:noise-random-committee}
\end{figure}

\begin{figure}
    \centering
    \includegraphics[width=\textwidth]{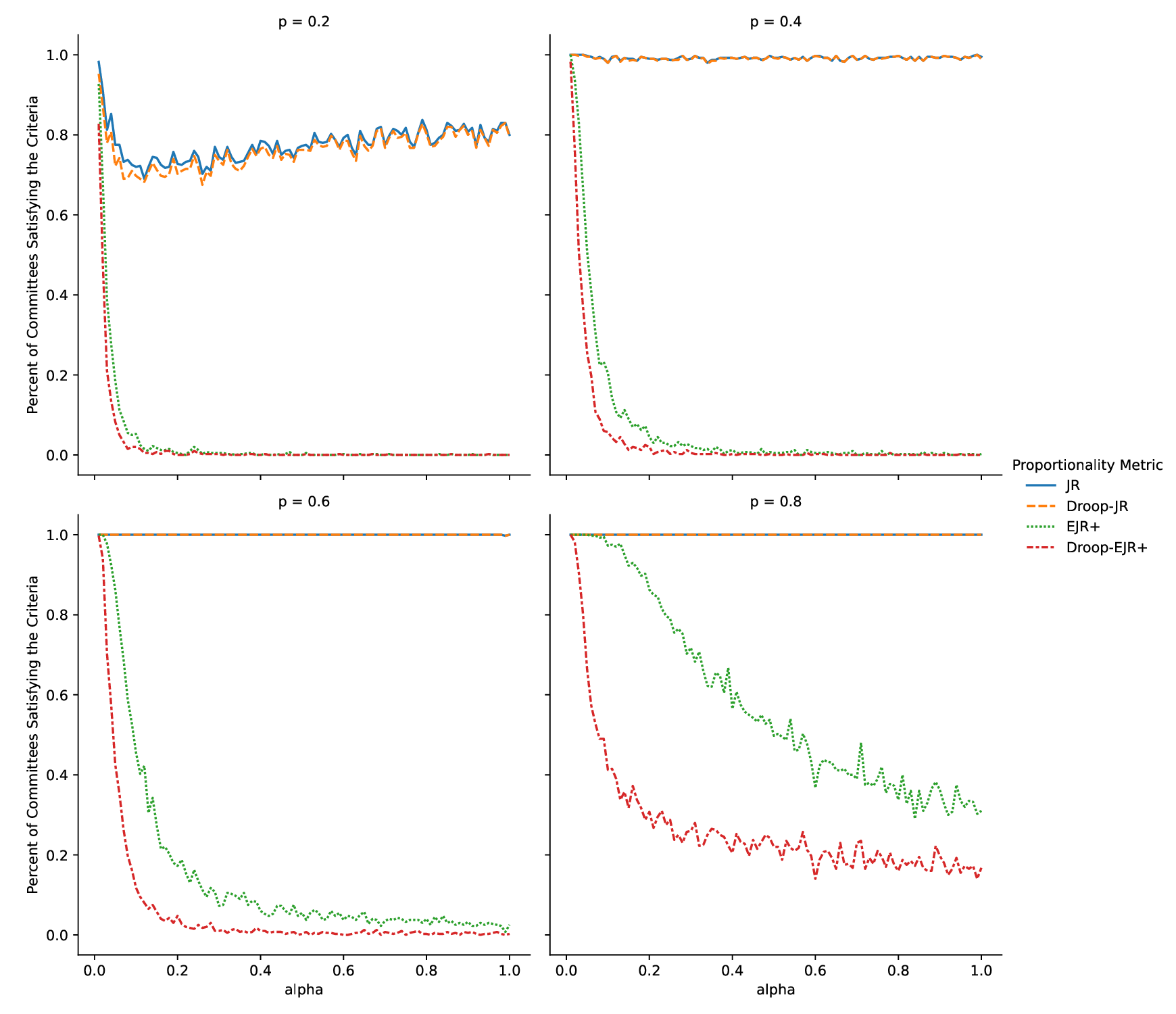}
    \caption{Results for the Truncated Urn Model with a randomly generated committee.}
    \label{fig:truncated-urn-random-committee}
\end{figure}

\begin{figure}
    \centering
    \includegraphics[width=\textwidth]{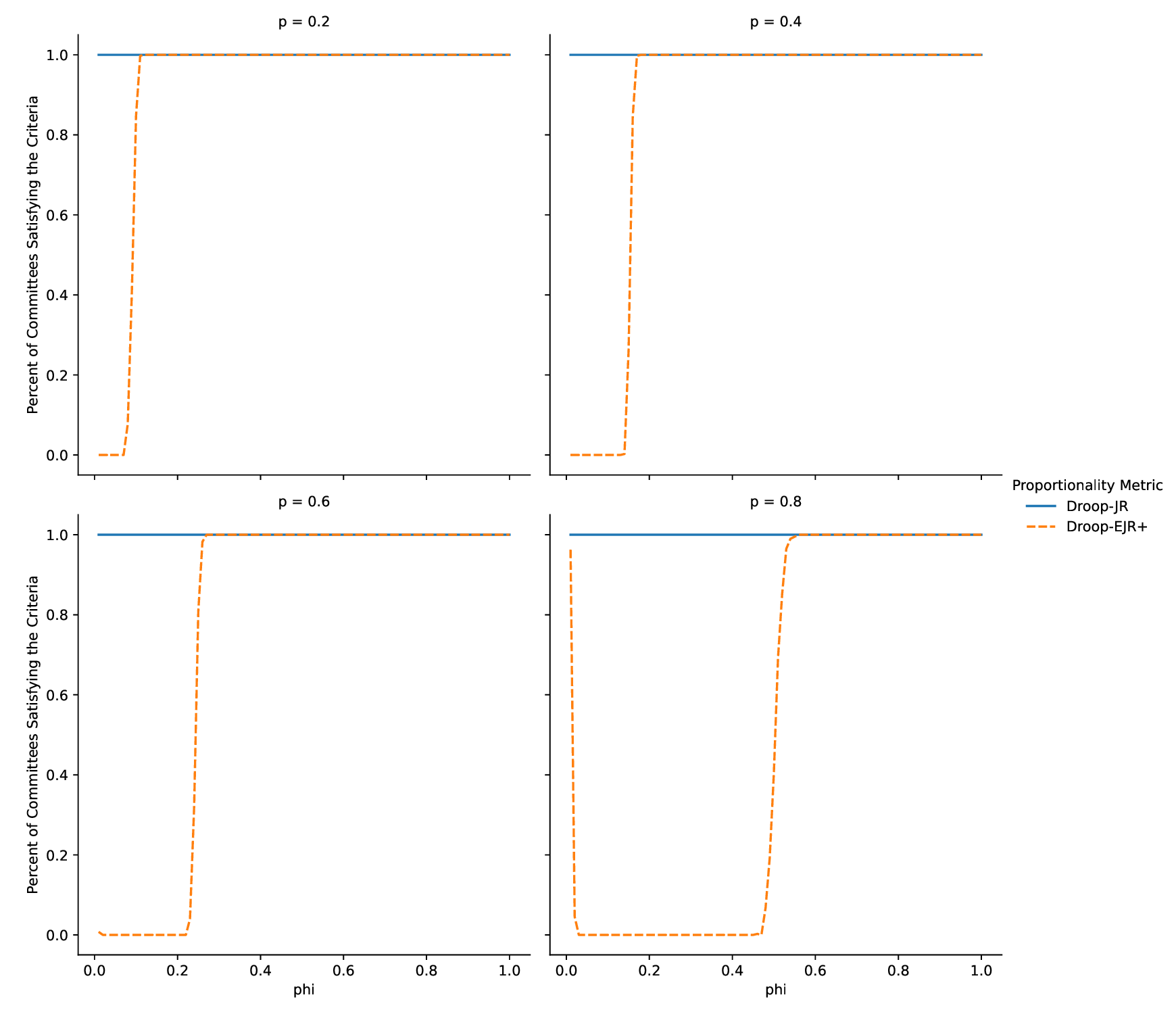}
    \caption{Results for the Resampling Model with an MES committee.}
    \label{fig:resampling-mes-committee}
\end{figure}

\begin{figure}
    \centering
    \includegraphics[width=\textwidth]{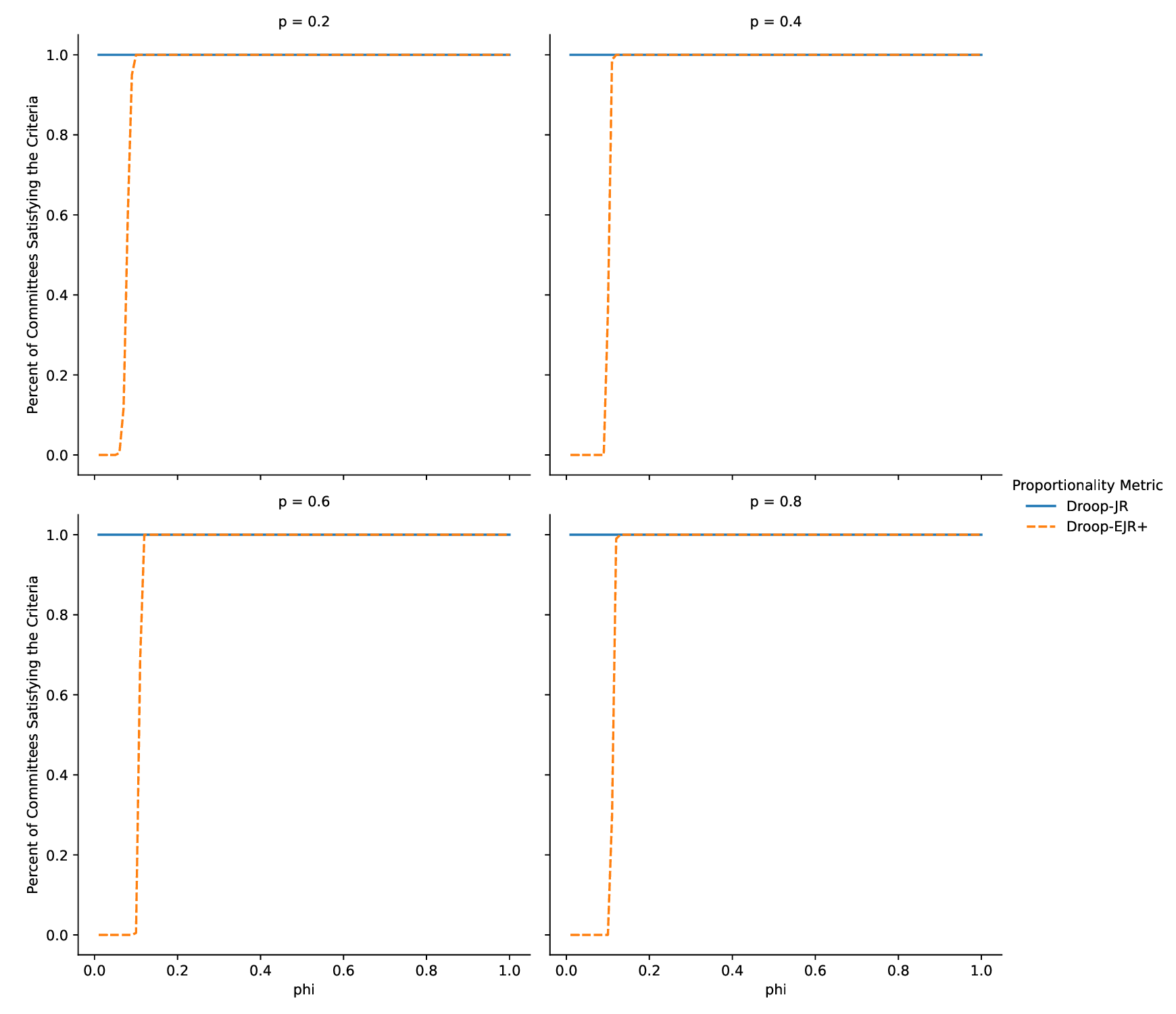}
    \caption{Results for the Noise Model with an MES committee.}
    \label{fig:noise-mes-committee}
\end{figure}

\begin{figure}
    \centering
    \includegraphics[width=\textwidth]{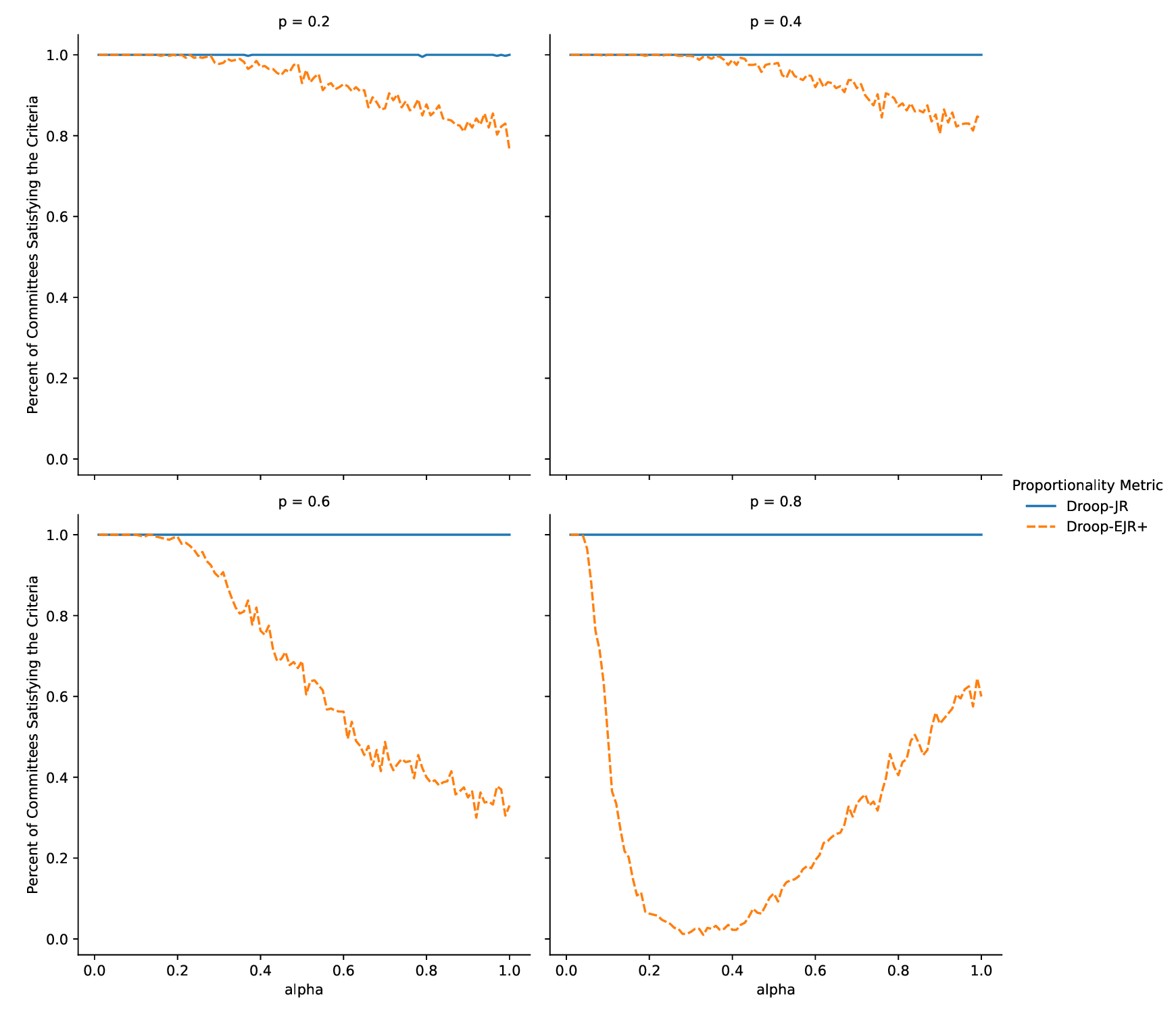}
    \caption{Results for the Truncated Urn Model with an MES committee.}
    \label{fig:truncated-urn-mes-committee}
\end{figure}

\begin{figure}
    \centering
    \includegraphics[width=\textwidth]{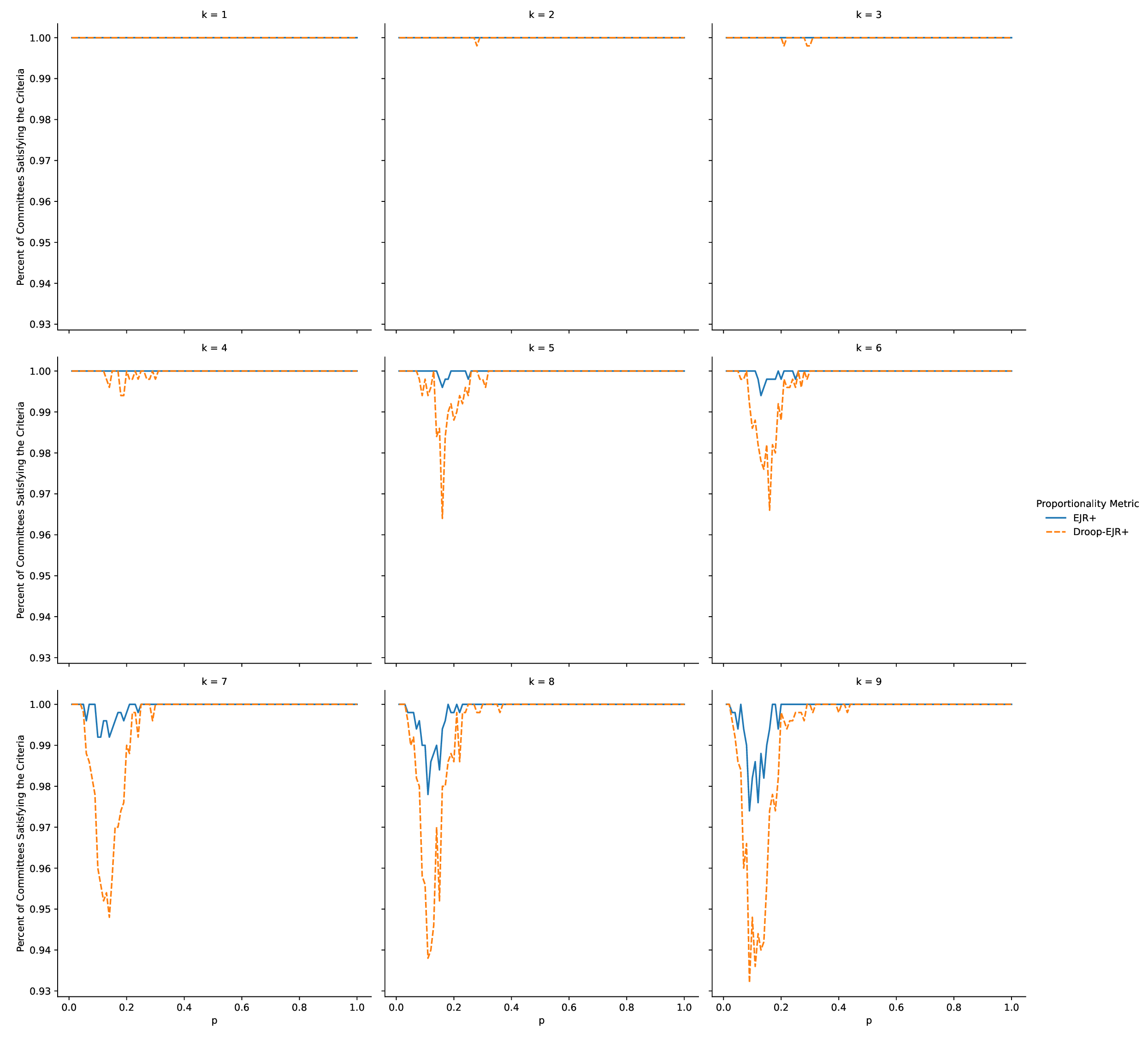}
    \caption{Results for the $p$-Impartial Culture model with 50 candidates.}
    \label{fig:IC-50}
\end{figure}

\begin{figure}
    \centering
    \includegraphics[width=\textwidth]{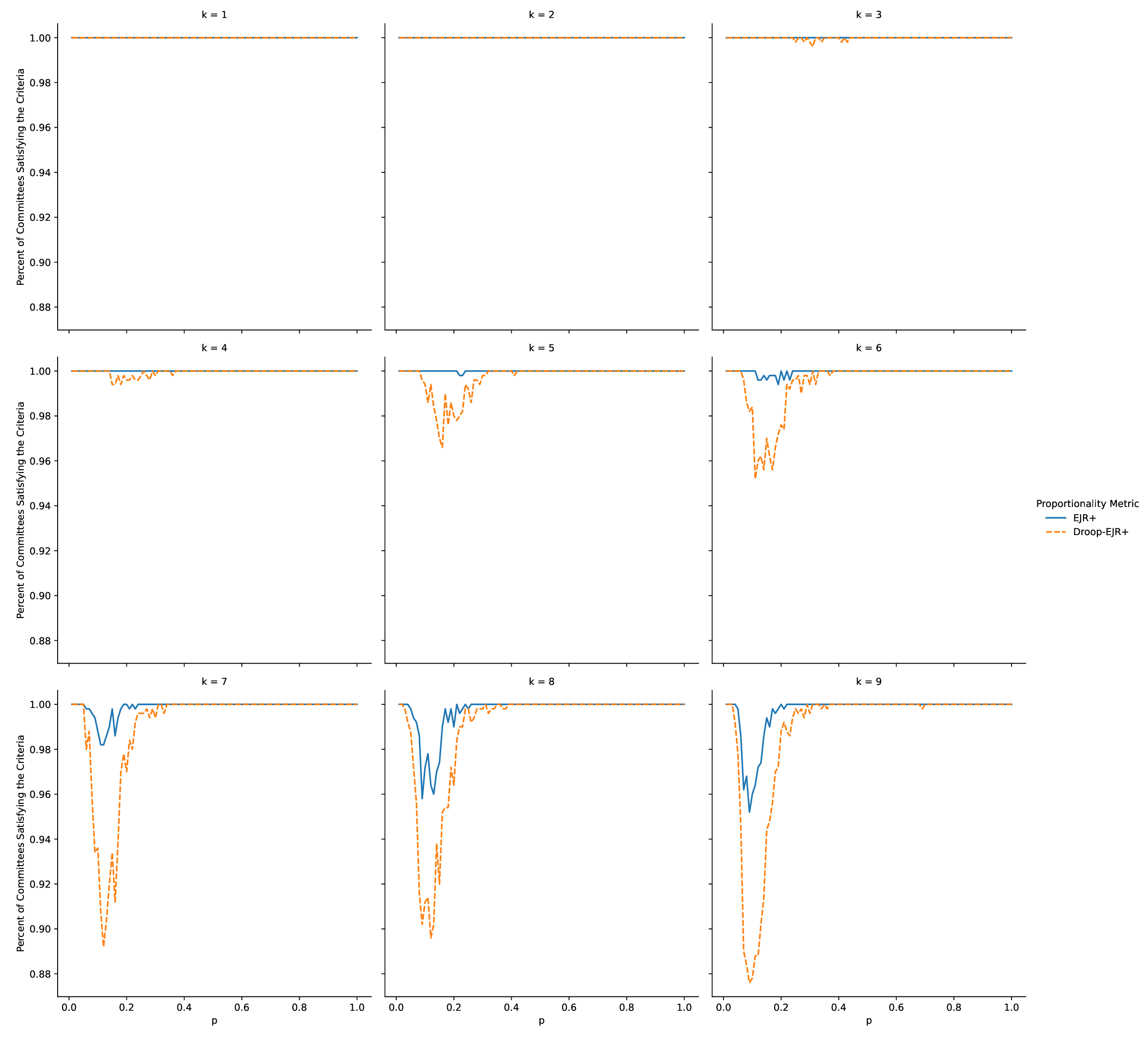}
    \caption{Results for the $p$-Impartial Culture model with 100 candidates.}
    \label{fig:IC-100}
\end{figure}

\begin{figure}
    \centering
    \includegraphics[width=\textwidth]{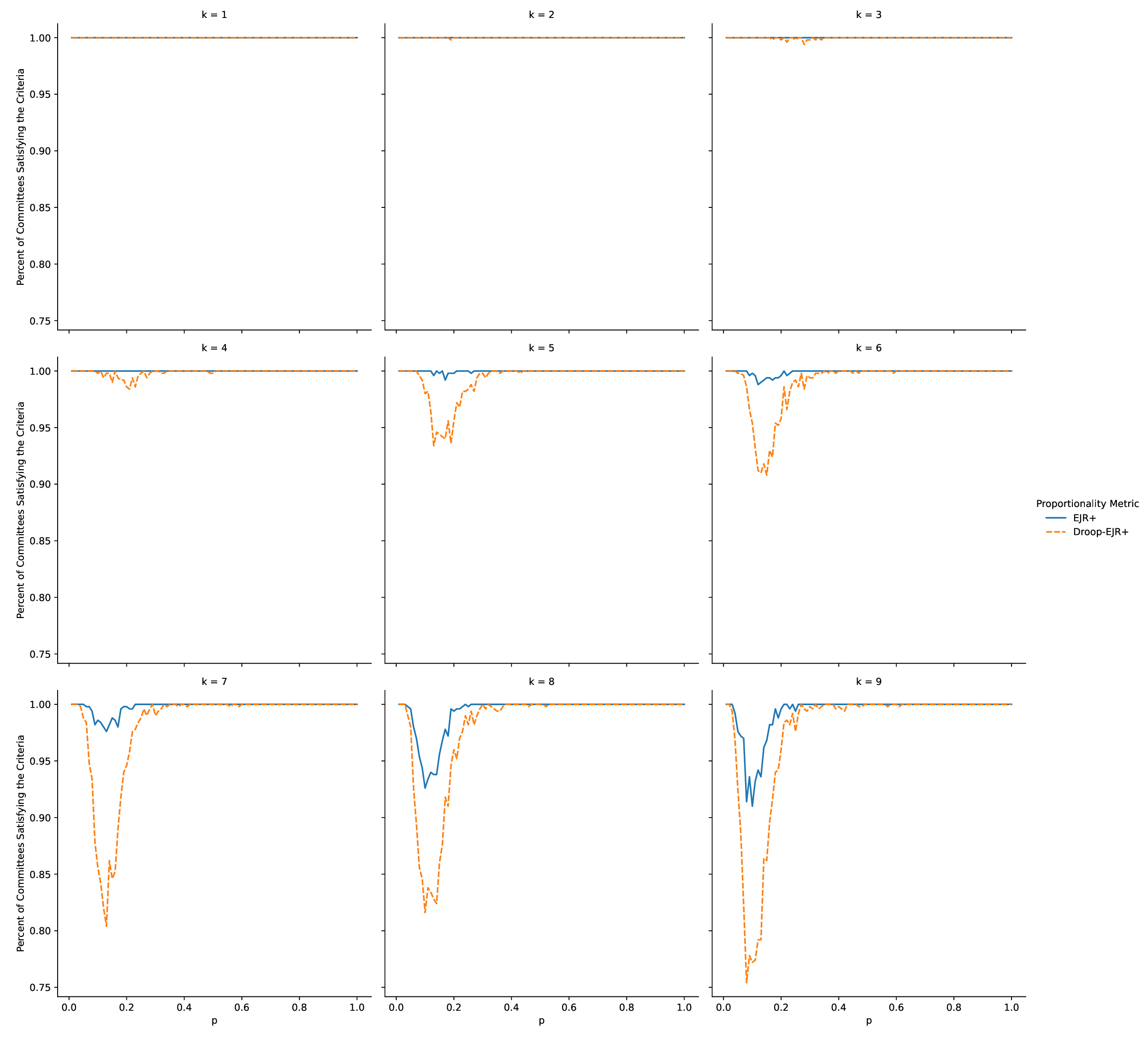}
    \caption{Results for the $p$-Impartial Culture model with 200 candidates.}
    \label{fig:IC-200}
\end{figure}
% Other potential additions to the paper:
% Add whatever implications are necessary to show that the Droop axioms are incomparable to the Hare axioms
% Proportionality degree/average satisfaction
% Check papers to see if any other big theorems that should be deproduced in Droop-land. (Hardness and various non satisfaction/nonequivalence results come to mind. Probably can just mention they also hold)
% Mention DroopAV (alternative to HareAV, studied in Aziz 2017)

% Extensions:
%JR for ranking based algorithms (solid quotas). Note that Brill/Peters shows that PJR+ <=> IPSC. Maybe some version of Droop-IPSC is interesting?
% PB?
%greedy capture rule and metric JR/proportional fairness/individual fairness for clustering setting. Also budgeted approvals

\end{document}